\pdfoutput=1
\documentclass[10pt, twoside]{article}

\usepackage[section]{algorithm}
\usepackage{amssymb,epsfig,amsmath,enumitem,mathrsfs,bm, cases,algpseudocode}
\usepackage{amsthm}
\usepackage{mathtools}
\usepackage{dsfont}
\usepackage{amsfonts}
\usepackage{natbib}
\usepackage[section]{placeins}
\usepackage{graphicx}
\usepackage{float}
\usepackage{caption}
\usepackage[aboveskip=0.5pt, belowskip=1.5pt]{subcaption}
\usepackage{color}
\definecolor{darkred}{rgb}{0.5,0.2,0.2}
\usepackage[pdftitle={State-dependent Hawkes processes and their application to limit order book modelling},pdfauthor={Maxime Morariu-Patrichi and Mikko S. Pakkanen},colorlinks=TRUE,allcolors=darkred]{hyperref}
\usepackage{eso-pic}
\usepackage[top=1.25in, bottom=1.25in, left=1in, right=1in]{geometry}
\usepackage[utf8]{inputenc}
\usepackage{datetime}
\usepackage{apptools}
\usepackage{titlesec}
\usepackage{booktabs}

\allowdisplaybreaks

\newdateformat{mydate}{\THEDAY\ \monthname[\THEMONTH] \THEYEAR}

\titleformat{\subsubsection}[runin]
  {\normalfont\normalsize\bfseries}{\thesubsubsection}{1em}{}

\makeatletter
\def\@endtheorem{\endtrivlist}
\makeatother

\graphicspath{{plots/}}

\usepackage{fancyhdr}
\fancypagestyle{firstpg}{%
	\fancyhf{}
	\fancyfoot[C]{\small\thepage}
	\fancyhead[CO]{}
	\fancyhead[CE]{}
	
}

\newcommand{\vect}[1]{\bm{#1}} 
\DeclareMathOperator*{\argmax}{arg\,max} 
\newcommand{\collection}[2]{(#1)_{#2}} 
\newcommand{\indicator}{\mathds{1}} 
\newcommand{\prob}{\mathbb{P}} 
\newcommand{\E}[1]{\mathbb{E}\left[ #1 \right]}  
\newcommand{\internalHistory}[2]{\mathcal{F}^{#1}_{#2}} 
\newcommand{\internalHistoryCollection}[1]{\mathbb{F}^{#1}}
\newcommand{\aHistory}[1]{\mathcal{F}_{#1}} 
\newcommand{\aHistoryCollection}{\mathbb{F}}
\newcommand{\intensityOfEvents}{\vect{\lambda}} 
\newcommand{\intensityOfEvent}[1]{\lambda_{#1}} 
\newcommand{\intensityOfEventsStates}{\vect{\tilde{\lambda}}}
\newcommand{\intensityOfEventState}[2]{\tilde{\lambda}_{#1#2}}
\newcommand{\hatintensityOfEventState}[2]{\hat{\tilde{\lambda}}_{#1#2}}
\newcommand{\integers}{\mathbb{N}} 
\newcommand{\realLine}{\mathbb{R}} 
\newcommand{\realLinePositive}{\mathbb{R}_{>0}} 
\newcommand{\realLineNonNegative}{\mathbb{R}_{\geq 0}} 
\newcommand{\numberOfEvents}{d_e} 
\newcommand{\eventSpace}{\mathcal{E}} 
\newcommand{\numberOfStates}{d_x} 
\newcommand{\stateSpace}{\mathcal{X}} 
\newcommand{\N}{\vect{N}} 
\newcommand{\NofEvent}[1]{N_{#1}} 
\newcommand{\Nhybrid}{\vect{\tilde{N}}} 
\newcommand{\NhybridOfEventState}[2]{\tilde{N}_{#1#2}} 
\newcommand{\baseRate}[1]{\nu_{#1}} 
\newcommand{\baseRates}{\vect{\nu}} 
\newcommand{\kernels}{\vect{k}} 
\newcommand{\kernelAtoB}[2]{k_{#1 #2}} 
\newcommand{\kernelAtoBHat}[2]{\hat{k}_{#1 #2}} 
\newcommand{\transitionProbabilities}{\vect{\phi}} 
\newcommand{\transitionProbabilityIfEvent}[1]{\vect{\phi}_{#1}} 
\newcommand{\transitionProbabilityIfEventFromAtoB}[3]{\phi_{#1}(#2, #3)}
\newcommand{\numberOfParams}{p}  
\newcommand{\kernelsParams}{\vect{\theta}}  
\newcommand{\likelihood}{\mathcal{L}} 
\newcommand{\hatTransitionProbabilityIfEventFromAtoB}[3]{\hat{\phi}_{#1}(#2, #3)}
\newcommand{\hatIntensityOfEvent}[1]{\hat{\lambda}_{#1}} 
\newcommand{\residualOfEvent}[2]{r^{#1}_{#2}} 
\newcommand{\residualOfEventState}[3]{\tilde{r}^{#1 #2}_{#3}} 
\newcommand{\impactCoeff}[3]{\alpha_{#1 #2 #3}}  
\newcommand{\impactCoeffHat}[3]{\hat{\alpha}_{#1 #2 #3}}  
\newcommand{\impactCoeffs}{\vect{\alpha}}
\newcommand{\decayCoeff}[3]{\beta_{#1 #2 #3}}  
\newcommand{\decayCoeffHat}[3]{\hat{\beta}_{#1 #2 #3}}  
\newcommand{\decayCoeffs}{\vect{\beta}}
\newcommand{\event}[1]{{\fontfamily{phv}\selectfont \small #1}}
\newcommand{\state}[1]{{\fontfamily{phv}\selectfont \small #1}}

\newcommand\mikko[1]{\textcolor{black}{#1}}
\newcommand\maxime[1]{\textcolor{black}{#1}}
\newcommand\mikkoagain[1]{\textcolor{black}{#1}}
\newcommand\mikkorev[1]{\textcolor{black}{#1}}

\begin{document}

\thispagestyle{firstpg}
\begin{center}
\large \textbf{STATE-DEPENDENT HAWKES PROCESSES AND\\THEIR APPLICATION TO LIMIT ORDER BOOK MODELLING} \\
	\vspace{1cm}
	\large  Maxime Morariu-Patrichi\footnote{Department of Mathematics, Imperial College London, South Kensington Campus, London SW7 2AZ, UK.} \hspace{1cm} Mikko S. Pakkanen\footnote{Department of Mathematics and Data Science Institute, Imperial College London, South Kensington Campus, London SW7 2AZ, UK and CREATES, Aarhus University, Aarhus, Denmark. E-mail: \href{mailto:m.pakkanen@imperial.ac.uk}{\texttt{m.pakkanen@imperial.ac.uk}} URL: \url{http://www.mikkopakkanen.fi}} $\mbox{\hspace{0cm}}$ \\
	\vspace{0.5cm}
	\normalsize {\mydate\today}
	\vspace{1cm}
\end{center}

\begin{abstract}
We \mikko{study statistical aspects of} state-dependent Hawkes processes, which are an extension of Hawkes processes where a \mikko{self- and cross-exciting} counting process and a state process are fully coupled, \mikko{interacting with each other}. \mikko{The excitation kernel of the counting process depends on the state process that, reciprocally, switches state when there is an event in the counting process.} We \mikko{first establish the} existence and uniqueness \mikko{of state-dependent Hawkes processes and explain how they can be simulated.} \mikko{Then we develop maximum likelihood estimation methodology for parametric specifications of the process. 
We apply state-dependent Hawkes processes to high-frequency limit order book data, allowing us to build a novel model that captures the feedback loop between the order flow and the shape of the limit order book. We estimate two specifications of the model, using the bid--ask spread and the queue imbalance as state variables, and find that excitation effects in the order flow are strongly state-dependent. Additionally, we find that the endogeneity of the order flow, measured by the magnitude of excitation, is also state-dependent, being more pronounced in disequilibrium states of the limit order book.}
\end{abstract}


\bigskip

\noindent \mikko{\textbf{MSC 2010 Classification:} 60G55, 62M09, 62P05}

\noindent \mikko{\textbf{JEL Classification:} C58, C51, G17}

\numberwithin{equation}{section}

\section{Introduction}

Hawkes processes are a class of \mikko{self- and cross-exciting} point processes \mikko{where} events of different types \mikko{may increase the rate of} \mikko{new events of the same or other type} \citep{Hawkes:1971aa, laub:2015:Hawkes}, \mikko{dispensing with} the \mikko{independence-of-increments} property of Poisson processes. Given their ability to capture clustering \mikko{and contagion effects}, \mikko{Hawkes processes} have found numerous applications in finance in the last decade, \mikko{most notably in the modelling of high-frequency financial data} \citep{Embrechts2011,Bacry:2015aa}. \mikko{Concurrently, the last two decades have witnessed a major transformation of financial markets through the proliferation of electronic trading in order-driven markets}, where \mikko{traders} submit buy and sell orders to \mikko{an electronic trading platform that attempts to match the arriving orders} \citep{gould2013surveylimit}. \mikko{A significant proportion of the order flow is driven by high-frequency trading algorithms, whereby successive orders are sometimes separated in time only by a few microseconds.} As billions of orders are submitted each day, this \mikko{has given rise to} a profusion of new \mikko{market} data to study, \mikko{giving an unprecedented opportunity} to \mikko{statistically analyse and model price formation and market microstructure} at the \mikko{shortest possible} timescales. 

\mikko{From the point of view of statistical modelling, a key objective here is to find an} accurate\mikko{, yet parsimonious,} statistical dynamic \mikko{description} of the \mikko{\emph{limit order book}} (LOB), that is, the \mikko{record of orders for a given asset that currently remain unfilled}  \citep{gould2013surveylimit}. As one of the main \mikko{approaches to LOB modelling}, Hawkes processes have been used \mikko{to describe} the order flow, which \mikko{inherently} drives the evolution of the LOB \citep{Large:2007aa:MeasureResiliency, bacry2016estimation, rambaldi:2016:volumeHawkes, Carteaetal2018, Lu2018}. \mikko{This approach proceeds by specifying a set} of order types (i.e., orders \mikko{classified in terms of their impact} on the LOB) and \mikko{by fitting} a \mikko{multivariate} Hawkes process to their timestamps. Hawkes processes \mikko{have essentially provided us with a prism through which we can see the market reaction to the different order types and their interaction.} More generally, the \mikko{Hawkes-process} approach fits into the \mikko{longer} trend of using point processes to model \mikko{sequences} of \mikko{irregularly-spaced market} events in continuous time \citep{Engle1998,Hautsch:book:2004, Bauwens:2009aa:HighFrequencyPoint, Bacry:2013:PriceHawkes, BacryMuzy:2014:HawkesPrice}.

\mikko{Whilst LOB models based on Hawkes processes have been rather successful in describing the dynamics of the order flow, they are unable to incorporate any endogenous \emph{state variables} describing the LOB, such as prices, volumes or the bid--ask spread, nor their} influence on the arrival rate of orders.
 We note that Hawkes processes have been used \mikko{to build full-fledged} LOB models \citep{MuniToke2011FullLOBModel, Abergel2015}, but \mikko{the arrival rate of orders in these models is not influenced by the state of the LOB}. In fact, Hawkes processes \mikko{are complemented by the other} main approach to \mikko{LOB modelling,} based on continuous-time Markov chains \citep{huang:2015:QueueReactive, huang:2015:ergodicit, Cartea:Donelly:2018}. \mikko{This alternative approach extends some of the earlier} zero-intelligence models \citep{smith:2003:continuousDoubleAuction, Cont2010, cont2013:priceDynamics} \mikko{and postulates that} the arrival rate of orders is \mikko{driven by the state of the LOB alone, making the state part of the model.} This has the merit of introducing a feedback loop between the order flow and the \mikko{state of the LOB}, but \mikko{it omits the self- and cross-excitation effects evidenced by the empirical work based on Hawkes processes.}

In short, each of these two approaches \mikko{to LOB modelling} has \mikko{desirable} qualities that the other lacks\mikko{, as also} noticed by \citet{bacry2016estimation}, \citet{Taranto2016}, \citet{Gonzales2017} and \citet{morariu:2017:hybrid}. \mikko{In fact,} \citet{Gonzales2017} \mikko{present empirical results, demonstrating} that the type of the next order depends on both the type of the previous order and the \mikko{state of the LOB}, which \mikko{highlights the need for a more general modelling framework that can amalgamate the two approaches.}

\mikko{The purpose of this paper is to introduce a novel, \emph{state-dependent} extension of a Hawkes process, in the context of statistical modelling. The new model, formulated in Section \ref{sec:state-hawkes}, pairs a multivariate} point process $\N$, \mikko{governed by a past-dependent stochastic intensity,} with \mikko{an observable} state process $X$. \mikkoagain{The intensity of the point process $\N$ is Hawkes process-like, characteristically depending on the past of $\N$, but also influenced by the past of $X$.} The process $X$ switches state when there is an event in $\N$, according to a Markov transition matrix that depends on the type of the event. This \mikkoagain{\emph{two-way}} interaction between $\N$ and $X$ makes them fully coupled, just like the order flow and \mikko{the state of the LOB in an order-driven market. It also distinguishes the model from the existing regime switching Hawkes processes \citep{Wang2012,Cohen:2013:regimeSwitchingHawkes,VinkovskayaEkaterina2014Appm,Swishchuk:2017:compoundHawkes}, where the state process evolves exogenously, receiving no feedback from the point process. Mathematically, we lift the pair $(\N, X)$ to a higher-dimensional ordinary point process, which allows us to invoke the theory of \emph{hybrid marked point processes} \citep{morariu:2017:hybrid} to establish the existence and uniqueness of non-explosive solutions and derive a simulation algorithm for the process.} 

\mikko{Further, in Section \ref{sec:MLE} we develop a maximum likelihood (ML) estimation method for parametric specifications of a state-dependent Hawkes process, extending the ML methodology for ordinary Hawkes processes pioneered by \citet{Ozaki1979a}. We find that the likelihood function of a state-dependent Hawkes process has very convenient separable form that makes it possible to estimate the transition probabilities of $X$ independently of the parameters of the point process $\N$, which is remarkable given that $\N$ and $X$ are fully coupled. The upshot of this useful property is that, as far as ML estimation is concerned, state-dependent Hawkes processes are no harder to estimate than ordinary multivariate Hawkes processes.}

\mikko{In Section \ref{sec:application} we finally apply the new model and methodology to high-frequency financial data, estimating  for the first time an LOB model that accommodates an explicit feedback loop between the state of the LOB and the order flow with self- and cross-excitation effects. More concretely, we estimate two specifications of a state-dependent Hawkes process on level-I LOB data on the stock of Intel (INTC) from Nasdaq, from June 2017 until May 2018, using the \emph{bid--ask spread} and \emph{queue imbalance}, respectively, as the state variable.} The latter \mikko{variable} can be interpreted as an indicator of the \mikko{shape of the LOB} and is a popular price-predictive signal \citep{Cartea:Donelly:2018}. \mikko{Our estimation results reveal} that the magnitude \mikko{and duration of excitation effects in the order flow depend} significantly on \mikko{both} state variables\mikko{, being effective at timescales ranging from 100 microseconds to 100 milliseconds.} \mikko{Since human reaction times have been measured to be of the order of 200 milliseconds, and longer for tasks involving choice \citep{Kosinski2006},} this is a clear sign of the \mikko{highly algorithmic nature of the order flow in modern electronic markets.} 
\mikko{Moreover, we find that the high level of endogeneity of the order flow, dubbed \emph{critical reflexivity} by \citet{Filimonov2012} and \citet{Hardiman2013}, and quantifiable here by the spectral norm of the Hawkes excitation kernel, is state-dependent and, intriguingly, more pronounced in what can be seen as \emph{disequilibrium states} of the LOB.}

In the \hyperref[sec:appendix]{Appendix} \mikko{of this paper}, we prove the \mikko{theoretical results of} Sections \ref{sec:state-hawkes} and \ref{sec:MLE} and \mikko{gather some auxiliary details on ML estimation. A Python library} \maxime{called \texttt{mpoints} that  implements the models and estimation methodology of this paper is available from \url{https://mpoints.readthedocs.io}.}

\section{State-dependent Hawkes processes} \label{sec:state-hawkes}

\subsection{\mikko{Preliminaries and definition}}

To \mikko{set the stage for} state-dependent Hawkes processes, we \mikko{first} review some central concepts of point process theory \citep{bremaud:1981:MartingaleDynamics, daleyVereJonesVolume1}. A $\numberOfEvents$-dimensional multivariate point process consists of an increasing sequence of \mikko{positive random} times $\collection{T_n}{n\in\integers}$ and a \mikko{matching} sequence of random marks $\collection{E_n}{n\in\integers}$ in $\eventSpace:=\{1,\ldots,\numberOfEvents\}$. \mikko{We interpret, for any $n \in \integers$, the pair $(T_n,E_n)$ as an event of type $E_n$ occurring at time $T_n$.} A point process can be \mikko{identified} by its counting process \mikko{representation} $\N:=(\NofEvent{1},\ldots,\NofEvent{\numberOfEvents})$, where
\begin{equation*}
	\NofEvent{e}(t) := \sum_{n\in\integers}\indicator_{\{T_n\leq t, E_n=e\}}, \quad t\geq 0, \quad e\in\eventSpace,
\end{equation*}
counts how many events of type $e$ have occurred \mikko{by} time $t$. \mikko{The counting process} $\N$ is said to be \mikko{\emph{non-explosive} if} $\lim_{n\rightarrow\infty}T_{n}=\infty$ with probability one. Given a filtration $\aHistoryCollection=\collection{\aHistory{t}}{t\geq 0}$ to which $\N$ is adapted, we say that a non-negative $\aHistoryCollection$-predictable process $\intensityOfEvents=(\intensityOfEvent{1},\ldots,\intensityOfEvent{\numberOfEvents})$ is the \mikko{\emph{$\aHistoryCollection$-intensity}} of $\N$ if
\begin{equation*}
	\E{\NofEvent{e}(t) - \NofEvent{e}(s) \,|\, \aHistory{s}} = \E{\int_s^t \intensityOfEvent{e}(u)du \,|\, \aHistory{s}},\quad t\geq s\geq 0, \quad e \in\eventSpace.
\end{equation*}
\mikko{Intuitively}, $\intensityOfEvent{e}(t)$ is the \mikko{infinitesimal} rate of \mikko{new} events of type $e$ at time $t$.

To \mikko{construct a} state-dependent Hawkes \mikko{process}, we \mikko{couple} the point process $\N$ with a c\`adl\`ag state process $\collection{X(t)}{t\geq 0}$ that takes values in a finite state space $\stateSpace:=\{1,\ldots,\numberOfStates\}$ and denote by $\internalHistoryCollection{\N, X}$ the natural filtration of the \mikko{pair} $(\N, X)$.
\theoremstyle{definition}
\newtheorem{state_dependent_hawkes}[algorithm]{Definition}
\begin{state_dependent_hawkes} \label{def:state_dependent_hawkes}
Let $\transitionProbabilities = \collection{\transitionProbabilityIfEvent{e}}{e\in\eventSpace}$ be a collection of $\numberOfStates\times\numberOfStates$ transition probability matrices, $\baseRates=(\baseRate{1},\ldots,\baseRate{\numberOfEvents})\in\realLinePositive^{\numberOfEvents}$ and $\kernels = \collection{\kernelAtoB{e'}{e}}{e',e\in\eventSpace}$, where $\kernelAtoB{e'}{e}:\realLinePositive\times\stateSpace\rightarrow\realLineNonNegative$.
The \mikko{pair} $(\N, X)$ is a \mikko{\emph{state-dependent Hawkes process}} with transition \mikko{distribution} $\transitionProbabilities$, base \mikko{rate vector} $\baseRates$ and \mikko{excitation kernel} $\kernels$\mikko{, abbreviated as $\mathrm{sdHawkes}(\transitionProbabilities,\baseRates,\kernels)$,} if
\begin{enumerate}[label=\textup{(\roman*)}]
	\item $\vect{N}$ \mikko{has} $\internalHistoryCollection{\vect{N}, X}$-intensity \mikko{$\intensityOfEvents$} that satisfies
	\begin{equation} \label{eq:event_intensity}
		\intensityOfEvent{e}(t) = \baseRate{e} + \sum_{e'\in\eventSpace}\int_{[0,t)}\kernelAtoB{e'}{e}(t-s, X(s))d\NofEvent{e'}(s) , \quad t\geq 0,\quad e\in\eventSpace;
	\end{equation}
	\item $X$ \mikko{is piecewise constant and} jumps only at the event times $\collection{T_n}{n\in\integers}$\mikko{, so that}
	\begin{equation}\label{eq:trans_prob}
		\prob\left(X(T_{n}) = x \,|\, E_n, \internalHistory{\N, X}{T_{n}-}\right) = \transitionProbabilityIfEventFromAtoB{E_n}{X(T_{n}-)}{x} ,\quad  n\in\integers, \quad x\in\stateSpace,
	\end{equation}
\end{enumerate} 
\mikko{where $\internalHistory{\N, X}{T_{n}-} := \bigvee_{\varepsilon > 0} \internalHistory{\N, X}{T_{n}-\varepsilon}$ and $X(T_{n}-) := \lim_{t \uparrow T_n} X(t)$.}
\end{state_dependent_hawkes}
\theoremstyle{plain}

\mikko{In applications to LOB modelling, t}he counting process $\N$ and state process $X$ will represent the order flow and the state of the \mikko{LOB}, respectively. \mikko{An event of type $e'\in\eventSpace$ at time $T=T_n$, for some $n \in \integers$, increases the infinitesimal rate of new events of type $e\in\eventSpace$ at time $t>T$ by $\kernelAtoB{e'}{e}(t-T,X(T))$.} The novelty here is that the self- and cross-excitation effects \mikko{now} depend on the state process $X$. \mikko{Reciprocally,} at each event in $\N$\mikko{, the state process $X$ may switch to a new state according to the probability \eqref{eq:trans_prob} that depends on the event type.}
Thereby, \mikko{mimicking the mechanics of the LOB}, \mikko{the processes} $\N$ and $X$ are fully coupled. \mikko{In the case} $\numberOfStates=1$, \mikko{Definition \ref{def:state_dependent_hawkes} reduces to that of an ordinary linear Hawkes process.}
 A \mikko{simulated} sample path of a state-dependent Hawkes process with $\numberOfEvents=1$, $\numberOfStates=2$ and exponential \mikko{kernel} is shown in Figure \ref{fig:simulation}. In this example, the process exhibits self-excitation only in the second state.

\theoremstyle{definition}
\newtheorem{non-Markov}[algorithm]{Remark}
\begin{non-Markov}
\begin{enumerate}[label=\textup{(\roman*)}]
\item \mikkoagain{The state space $\stateSpace$ is chosen to be finite here on the grounds of practical estimation. Theoretically, a state-dependent Hawkes process can however be defined in more general (infinite) state spaces \citep[Example 2.15]{morariu:2017:hybrid}.}
\item \mikkoagain{The relationship \eqref{eq:trans_prob} does not imply that $X$ is a pure jump Markov process. In particular, the time to next state transition need not be conditionally exponentially distributed, since it is determined by the counting process $\N$, which nests a multitude of non-Markovian ordinary Hawkes processes, for instance.}
\item \mikkoagain{If for any $e \in \eventSpace$ the transition distribution $\big(\transitionProbabilityIfEventFromAtoB{e}{x'}{x}\big)_{x \in \stateSpace}$ does not depend on the previous state $x'\in \stateSpace$, the process $\N$ reduces to a \emph{marked Hawkes process}, as described in \citet[Section 2.2.1]{Bacry:2015aa}.} 
\end{enumerate}
\end{non-Markov}

\begin{figure}[h]
  \centering
    \includegraphics[width=0.6\textwidth]{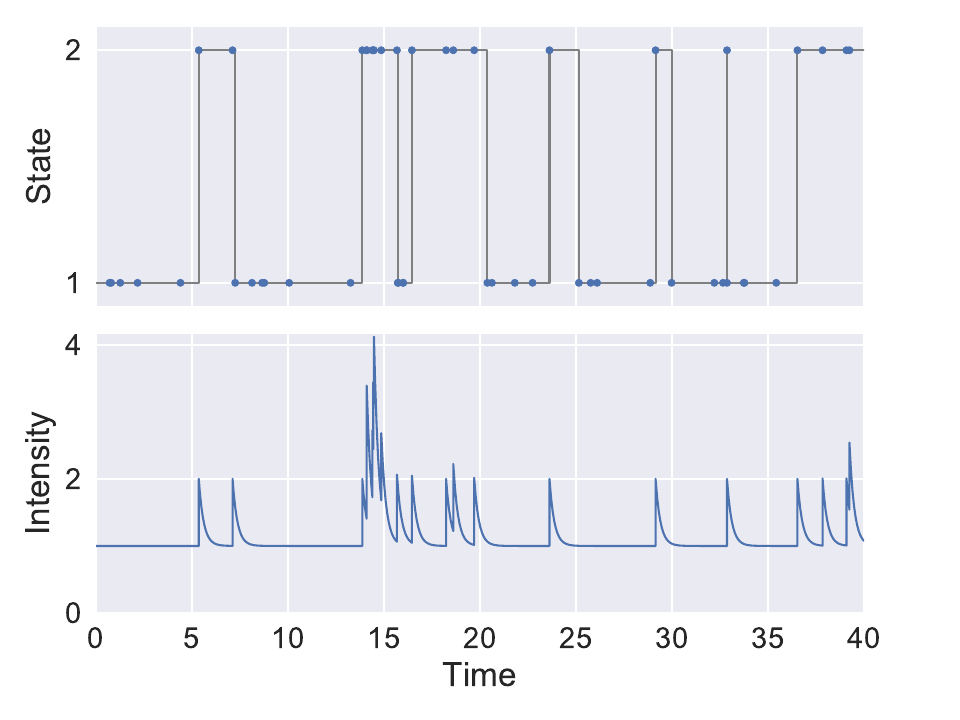}
    \caption[Simulation of a state-dependent Hawkes process]{Simulation of a state-dependent Hawkes process with $\numberOfEvents = 1 $, $\numberOfStates =2$. The upper plot shows the evolution of the state process. The blue dots indicate the event times and the lower plot represents the intensity. \mikko{The process is specified so that $\nu_1 = 1$ and $k_{11}(t,x) = \exp(-\maxime{4} t)\indicator_{\{x=2\}}$, that is, i}n state 2 the process exhibits \mikko{exponential} self-excitation whereas no self-excitation occurs in state 1.}
	\label{fig:simulation}
\end{figure}

\subsection{\mikko{Existence and uniqueness of non-explosive solutions}} \label{subsec:existence}

\mikko{Since Definition \ref{def:state_dependent_hawkes} is implicit in the sense that the counting process $\N$ is defined by intensity $\intensityOfEvents$ that depends on the past of $\N$ and $X$, care is needed to establish the existence and uniqueness of a pair $(\N,X)$ that solves \eqref{eq:event_intensity} and \eqref{eq:trans_prob} so that $\N$ is non-explosive. To this end, we lift $(\N, X)$ to} a $\numberOfEvents\numberOfStates$-dimensional multivariate point process $\Nhybrid=\collection{\NhybridOfEventState{e}{x}}{e\in\eventSpace, x\in\stateSpace}$, where $\NhybridOfEventState{e}{x}$ counts the number of events of type $e$ after which the state is $x$, \mikko{formally,}
\begin{equation*}
\mikko{\NhybridOfEventState{e}{x}(t) = \sum_{n\in\integers}\indicator_{\{T_n\leq t, E_n=e, X(T_n)=x\}}, \quad t\geq 0, \quad e\in\eventSpace, \quad x\in\stateSpace.}
\end{equation*}
\mikko{The marks corresponding to $\Nhybrid$ are given by $(E_n, X_n)$, $n \in \integers$, where $E_n$ is as above and} $X_n := X(T_n)$ is the value of the state process following the \mikko{$n$th} event and $X_0$ is the initial state.
Thus, given $\Nhybrid$, \mikko{the state} $X(t)$ can be recovered from the most recent mark at time $t$\mikko{, which makes the relationship between $\Nhybrid$ and $(\N,X)$ bijective.} \mikko{In fact, state-dependent Hawkes processes were first introduced using this representation in \citet{morariu:2017:hybrid}.}

By applying the general characterisation result in \citet{morariu:2017:hybrid}, the dynamics of $\Nhybrid$ can be expressed in terms of the dynamics of $(\N, X)$, and vice versa. The natural filtration of $\Nhybrid$ is denoted by $\internalHistoryCollection{\Nhybrid}$.
\newtheorem{characterisation}[algorithm]{Theorem}
\begin{characterisation} \label{thm:characterisation}
\mikko{The pair} $(\N, X)$ is a non-explosive \mikko{$\mathrm{sdHawkes}(\transitionProbabilities,\baseRates,\kernels)$ process} if and only if $\Nhybrid$ is non-explosive, \mikko{admitting} $\internalHistoryCollection{\Nhybrid}$-intensity $\intensityOfEventsStates$ that satisfies
\begin{equation} \label{eq:hybrid_intensity}
	\intensityOfEventState{e}{x}(t) = \transitionProbabilityIfEventFromAtoB{e}{X(t)}{x}\left( \baseRate{e} + \sum_{e'\in\eventSpace, x'\in\stateSpace}\int_{[0,t)}\kernelAtoB{e'}{e}(t-s, x')d\NhybridOfEventState{e'}{x'}(s) \right),\quad t\geq 0, \quad e\in\eventSpace,\quad x\in\stateSpace.
\end{equation}
\end{characterisation}

The above theorem in fact shows that state-dependent Hawkes processes belong to the class of \mikko{\emph{hybrid marked point processes}} studied in \citet{morariu:2017:hybrid}. The general existence and uniqueness results therein \mikko{apply to the present class of processes} as follows.
\newtheorem{existence_uniqueness}[algorithm]{Theorem}
\begin{existence_uniqueness} \label{thm:existence_uniqueness}
\mikko{A unique, non-explosive $\mathrm{sdHawkes}(\transitionProbabilities,\baseRates,\kernels)$ process exists} if one of the following two conditions is satisfied:
\begin{enumerate}[label=\textup{(\roman*)}]
	\item\label{item:bounded} the \mikko{components of} $\kernels$ are bounded \mikko{functions};
	\item\label{item:integral} $\sum_{e'\in\eventSpace, x'\in\stateSpace}\int_{0}^{\infty}\kernelAtoB{e'}{e}(t,x')dt < (\max_{x'\in\stateSpace}\transitionProbabilityIfEventFromAtoB{e}{x'}{x})^{-1}$ for all $e\in\eventSpace$ \mikko{and} $x\in\stateSpace$.
\end{enumerate}
\end{existence_uniqueness}

\mikko{Condition \ref{item:bounded} above suffices for the purposes of the present paper, where we apply bounded, exponential kernels. However, condition \ref{item:integral} is included here for completeness as it reduces to the classical stability condition of \citet{Massoulie:1998aa:StabilityResults} in the case $\numberOfStates=1$.}

\subsection{Simulation} \label{subsec:simulation}

Another implication of Theorem \ref{thm:characterisation} is that the simulation of a state-dependent Hawkes process can be reduced to the simulation of a multivariate point process with an intensity given by \eqref{eq:hybrid_intensity} and, thus, many simulation techniques from point process theory can be reused \citep{Lewis1976:SimulationPoissonProcesses, daleyVereJonesVolume1}.

In fact, the sample path in Figure \ref{fig:simulation} \mikko{has been} generated using \mikko{\emph{Ogata's thinning algorithm}} \citep{Ogata:1981aa}\mikko{, which is an exact simulation algorithm and is adaptable for state-dependent Hawkes processes as follows. We write $R(t):=\sum_{e\in\eventSpace, x\in\stateSpace}\intensityOfEventState{e}{x}(t)  = \sum_{e\in\eventSpace}\intensityOfEvent{e}(t)$, which is a function of all $(T_{n},E_{n},X_{n})$ such that $T_n < t$.}

\begin{algorithm}[!h]
\caption{Iterative step in Ogata's thinning algorithm for state-dependent Hawkes processes}\label{alg:Ogata}
\begin{algorithmic}[1]
\Require $(T_{i},E_{i},X_{i})_{i=1,\ldots,n-1}$
\State set $T := T_{n-1}$
\State set $\xi := 0$
\While{$\xi = 0$} 
\State draw $U \sim \mathrm{Exp}(R(T))$
\State set $\xi := 1$ with probability $\frac{R(T+U)}{R(T)}$
\State set $T := T+U$
\EndWhile
\State set $T_n := T$
\State draw $E_n \in \eventSpace$ with probabilities proportional to $\collection{\intensityOfEvent{e}(T_n)}{e\in\eventSpace}$ 
\State draw $X_n \in\stateSpace$ with probabilities $\collection{\transitionProbabilityIfEventFromAtoB{E_n}{X_{n-1}}{x}}{x\in\stateSpace}$
\State \Return $(T_n,E_n,X_n)$
\end{algorithmic}
\end{algorithm}

\theoremstyle{definition}
\newtheorem{simulation}[algorithm]{Remark}
\begin{simulation} \label{rem:simulation}
\begin{enumerate}[label=\textup{(\roman*)}]
\item \mikko{In Algorithm \ref{alg:Ogata},} we are implicitly assuming that the \mikko{components of the} kernel $\kernels$ \mikko{are} non-increasing, which guarantees that $R(t')\leq R(t)$ for all $t'\in[t,T_n]$.  In the general case, one needs to define $R(t)$ \mikko{so} that it bounds the total intensity $\sum_{e\in\eventSpace}\intensityOfEvent{e}(t')$ for all $t'\in[t,T_n]$.
\item \mikko{Lines 9--10 of Algorithm \ref{alg:Ogata} use the product form~\eqref{eq:hybrid_intensity} to simulate the marks, which avoids} the computation of $\numberOfEvents\numberOfStates$ products at the cost of generating an additional random number.
\end{enumerate}
\end{simulation}
\theoremstyle{plain}

\subsection{\mikkorev{Comparison with related models of limit order books}}
\mikkorev{We will now briefly compare state-dependent Hawkes processes, in the sense of Definition \ref{def:state_dependent_hawkes}, to some recent, closely related models, which similarly aim to couple a point process to a state process in the context of LOB modelling.}

\mikkorev{
The \emph{regime-switching Hawkes process} of \citet{VinkovskayaEkaterina2014Appm} can be seen as a special case of Definition \ref{def:state_dependent_hawkes}, with the exception that the dynamics of the state process $X$ are not modelled. (Effectively, this means that the counting process $\N$ is then specified conditional on a realisation of $X$, precluding two-way interaction between $\N$ and $X$.) In her empirical application, \citet{VinkovskayaEkaterina2014Appm} 
estimates a four-dimensional version of the model for arriving orders of four different types in the New York Stock Exchange Trades and Quotes (TAQ) data. She employs the bid--ask spread as a state variable, like we do in one of our models in Section~\ref{sec:application}.}

\mikkorev{
\citet{Swishchuk:2017:compoundHawkes} construct a \emph{compound Hawkes process} using a one-dimensional ordinary Hawkes process $N$ (possibly non-linear) and a Markov chain $(Y_n)_{n \in \mathbb{N}}$ in a finite state space $\mathcal{Y}$, independent of $N$. The actual process $(S(t))_{t \geq 0}$, which is used as a model of prices, is given by  
\begin{equation*}
S(t) := S(0) + \sum_{n=1}^{N(t)} a(Y_n), \quad t\geq 0, 
\end{equation*}
where $a$ is a function from $\mathcal{Y}$ to $\mathbb{R}$. Whilst having similar ingredients, the compound Hawkes process is otherwise unrelated to state-dependent Hawkes processes. In particular, the Markov chain $(Y_n)_{n \in \mathbb{N}}$ in the compound Hawkes process does not influence the rate of events, that is, price changes. \citet{Swishchuk:2017:compoundHawkes} establish limit theorems for the process, decribing its long-term behaviour, and also empirically estimate it using Nasdaq Stock Market LOB data.}

\mikkorev{
Note that in Definition \ref{def:state_dependent_hawkes} the excitation effect of each event is predicated upon the state prevailing at the time the event occurs. Alternatively, we could make the level of excitation track the current state. We could also make the base rates track the current state. Applying these modifications to equation \eqref{eq:event_intensity} in Definition \ref{def:state_dependent_hawkes} yields an analogous intensity satisfying
	\begin{equation} \label{eq:event_intensity_g2}
		\intensityOfEvent{e}(t) = \baseRate{e}(X(t-)) + \sum_{e'\in\eventSpace}\int_{[0,t)}\kernelAtoB{e'}{e}(t-s, X(t-))d\NofEvent{e'}(s) , \quad t\geq 0,\quad e\in\eventSpace.
	\end{equation}
Whether it is ultimately more natural to use $X(s)$ or $X(t-)$ in the excitation kernel is open to debate, but in the current context of LOB modelling, where the most significant excitation effects tend to be ephemeral, the difference is unlikely to be large in practice.}

\mikkorev{
Now, several recent LOB models conform to \eqref{eq:event_intensity_g2}. \citet{Cohen:2013:regimeSwitchingHawkes} introduce a one-dimensional \emph{Markov-modulated Hawkes process} following \eqref{eq:event_intensity_g2}, where $X$ is an exogenous Markov process in a finite state space (exogenous in the sense that the counting process does not influence $X$). The key feature of their work is that they assume $X$ to be unobservable, leading them to derive a filtering procedure for the estimation of the current state of $X$. \citet{Cohen:2013:regimeSwitchingHawkes} illustrate their methodology by estimating a regime switch in TAQ trade data during the US equity market flash crash on 6 May 2010.} 

\mikkorev{
Coinciding with the first preprint version of the present paper, \citet{Wuetal2019} develop a \emph{queue-reactive Hawkes process} based on \eqref{eq:event_intensity_g2}. In their model, $X$ is endogenous and carries information about queue lengths in the LOB, while the multi-dimensional counting process driven by the intensity \eqref{eq:event_intensity_g2} models events pertaining to these queues. \citet{Wuetal2019} estimate their model on German bond (Bund) and index (DAX) futures LOB data. Subsequently, \citet{Mounjid:2019:AsymptoticProperties} generalise the queue-reactive Hawkes process to a more general point process framework that allows for non-linearity and quadratic Hawkes structure. \citet{Mounjid:2019:AsymptoticProperties} additionally establish ergodicity for the model and also derive functional limit theorems for its long-term behaviour. They apply the model to evaluate and rank equities market makers on Euronext Paris.}

\mikkorev{
Finally, \citet{Fossetetal2020} develop a Hawkes process model of endogenous liquidity crises. They model the increase and decrease, respectively, in the bid--ask spread as a two-dimensional point process that conforms to \eqref{eq:event_intensity_g2}. The first component of the process follows an ordinary Hawkes process (possibly non-linear), whilst the second component has a state-dependent base rate, depending on the the current bid--ask spread. \citet{Fossetetal2020} then analyse the long-term stability of the model.}

\section{Parametric estimation via maximum likelihood} \label{sec:MLE}

	\mikko{Estimating the base rate vector} $\baseRates$ and kerne\mikko{l} $\kernels$ \mikko{of an ordinary Hawkes process has  become a vibrant} research topic in statistics and statistical finance literature. \mikko{Whilst the recent focus has mostly been on non-parametric methodology} \citep{Bacry2016, Kirchner:2017:estimation:Hawkes, Eichler2017, Sancetta2017a, Achab:2018aa}, \mikko{our aim in this paper is to extend} the classical parametric framework \citep{Ozaki1979a}, comprehensively summarised in  \citet{Bowsher2007}, to state-dependent Hawkes processes\mikko{, leaving non-parametric methodology for future work.} From now on, \mikko{we work with a kernel $\kernels = \kernels^{\kernelsParams}$} parametrised by a vector $\kernelsParams\in\realLine^{\numberOfParams}$.

\subsection{Likelihood \mikko{function}}

We know from Subsections \ref{subsec:existence} and \ref{subsec:simulation} that a state-dependent Hawkes process $(\N, X)$ can be \mikko{lifted to} a $\numberOfEvents\numberOfStates$-dimensional point process $\Nhybrid$. Given a realisation $\collection{t_n, e_n, x_n}{n=1,\ldots, N}$ of $\Nhybrid$ over a time horizon $[0, T]$, the likelihood \mikko{function} $\likelihood(\transitionProbabilities, \baseRates, \kernelsParams)$ \mikko{can be informally understood as} the probability that $\sum_{e\in\eventSpace, x\in\stateSpace}\NhybridOfEventState{e}{x}(T)=N$ and $\collection{T_n, E_n, X_n}{n=1,\ldots, N}$ \mikko{lies} in a small neighbourhood of  $\collection{t_n, e_n, x_n}{n=1,\ldots,N}$, under the assumption that $\collection{T_n, E_n, X_n}{n\in\integers}$ is generated by a state-dependent Hawkes processes with parameters $(\transitionProbabilities, \baseRates, \kernelsParams)$. \mikko{More rigorously, the} likelihood \mikko{function} is \mikko{the} density of the \mikko{\emph{Janossy measure}} with respect to the Lebesgue measure on $\realLine^{N}$ \citep[p.~125, 213]{daleyVereJonesVolume1}.

For \mikko{ordinary} Hawkes processes, \mikko{the likelihood function} $\likelihood(\baseRates, \kernelsParams)$ can be expressed directly in terms of $\intensityOfEvents$ \citep{daleyVereJonesVolume1} and \mikko{the maximum likelihood (ML) estimator} $(\hat \baseRates, \hat\kernelsParams)$ \mikko{is obtained by maximising} $\likelihood(\hat \baseRates, \hat\kernelsParams)$\mikko{, in practice numerically}. For state-dependent Hawkes processes, we are able to express $\likelihood(\transitionProbabilities, \baseRates, \kernelsParams)$ in terms of $\transitionProbabilities$ and $\intensityOfEvents$, and find that maximising the likelihood is \mikko{conveniently} achieved by solving two independent optimisation problems.
\newtheorem{likelihood_separability}[algorithm]{Theorem}
\begin{likelihood_separability} \label{thm:likelihood_separability}
	The \mikko{log likelihood function of an $\mathrm{sdHawkes}(\transitionProbabilities,\baseRates,\kernels^{\kernelsParams})$ process} is given by
	\begin{equation} \label{eq:log-likelihood}
		\ln \likelihood(\transitionProbabilities, \baseRates, \kernelsParams) = \sum_{n=1}^{N}\ln\transitionProbabilityIfEventFromAtoB{e_n}{x_{n-1}}{x_n} + \sum_{n=1}^n\ln\intensityOfEvent{e_n}(t_n) - \int_{0}^{T}\sum_{e\in\eventSpace}\intensityOfEvent{e}(t)dt.
	\end{equation}
	Furthermore, $(\hat\transitionProbabilities, \hat\baseRates, \hat\kernelsParams )\in\argmax_{\transitionProbabilities,\baseRates,\kernelsParams }\likelihood(\transitionProbabilities,\baseRates,\kernelsParams )$ if and only if
	\begin{numcases}{}
		\hatTransitionProbabilityIfEventFromAtoB{e}{x}{x'}=
		\frac{\sum_{n=1}^{N}\indicator(x_{n-1}=x, e_n=e, x_n=x')}{\sum_{n=1}^{N}\indicator(x_{n-1}=x, e_n=e)}, \quad e\in\eventSpace,\quad x,x'\in\stateSpace,\nonumber \\
		(\hat \baseRates, \hat \kernelsParams )\in \argmax_{\baseRates,\kernelsParams } \sum_{n=1}^{N}\ln \intensityOfEvent{e_n}(t_n) - \int_0^T\sum_{e\in\eventSpace}\intensityOfEvent{e}(t)dt.\label{eq:optimisation_problem}
	\end{numcases}
\end{likelihood_separability}

The upshot of Theorem \ref{thm:likelihood_separability} is that the ML estimation of state-dependent Hawkes processes is no harder than that of ordinary Hawkes processes.
\mikko{Namely,} $\transitionProbabilities$ is estimated \mikko{in a straightforward manner} by the empirical transition probabilities\mikko{,} whilst $(\baseRates,\kernelsParams)$ is estimated by maximising the \mikko{log} quasi-likelihood of $\N$, that is, the Radon-Nikodym derivative of a change of measure that transforms a standard Poisson process into $\N$ \citep[Theorem 10, p.~241]{bremaud:1981:MartingaleDynamics}\mikko{, which is similar to the log likelihood of a multivariate ordinary Hawkes process.} It is remarkable that, in spite of the strong coupling between the events and the state process, \mikko{the estimation of $\transitionProbabilities$ and $(\baseRates,\kernelsParams)$ is decoupled due to the separable form \ref{eq:log-likelihood} of the log likelihood function.}

\theoremstyle{definition}
\newtheorem{likelihood_formula}[algorithm]{Remark}
\begin{likelihood_formula} \label{rem:likelihood_formula}
\mikkoagain{It should be stressed that the separable form \eqref{eq:log-likelihood} of the log likelihood function $\likelihood$ is not a trivial consequence of} Definition \ref{def:state_dependent_hawkes}. It is once again the \mikko{lift} $\Nhybrid$ that allow us to transpose the problem to the setting of point process theory, where classical \mikko{results} \mikkoagain{apply} (see the proof of Theorem \ref{thm:likelihood_separability} in Appendix \ref{sec:proofs}).
\end{likelihood_formula}
\theoremstyle{plain}

We note that, in the case of ordinary Hawkes processes ($\numberOfStates=1$), consistency and asymptotic normality results for \mikko{the ML estimator} $(\hat \baseRates, \hat\kernelsParams)$ are available in the literature, see for example \citep{Ogata1978, Clinet2016}. However, these results rely on the stationarity and ergodicity of the underlying process and, unfortunately, \mikko{these properties need not extend to} state-dependent Hawkes processes \mikko{in general}. \mikkorev{If the excitation kernel $\kernels$ is dominated in all states by a non-state-dependent kernel that satisfies a classical stability condition \citep{Massoulie:1998aa:StabilityResults}, we conjecture that then stationarity and ergodicity hold. This can potentially be shown by adapting the arguments used in \citet{Mounjid:2019:AsymptoticProperties} to prove ergodicity for closely related point processes, or in the special case of exponential kernels (see Subsection \ref{subsec:exp_kernels}) using Markov process techniques, as outlined in Section 4.7 in \citet{Morariu:2019:thesis}. The case where $\kernels$ fails  to be dominated  by a stable kernel in some of the states appears however to be empirically more relevant given our results in Subsection \ref{sec:endo}. In this case, stationarity and ergodicity become more elusive, whilst they may still plausibly hold as long as the sojourns of the process in the ``unstable states'' do not dominate its time evolution. All in all, the analysis of the asymptotic properties of the ML estimator is unfortunately not straightforward and remains beyond the scope of the present paper.}

In Appendix \ref{subsec:test_estimation}\mikkoagain{,} we present Monte Carlo results that exemplify the \mikkoagain{favourable} finite-sample performance of the ML estimator. In particular, the results provide evidence that $(\hat\transitionProbabilities, \hat\baseRates, \hat\kernelsParams)$ is consistent as the length of the estimation window increases.

\subsection{\mikko{Goodness-of-fit diagnostics using residuals}}

\mikko{The goodness of fit of an estimated $\mathrm{sdHawkes}(\hat\transitionProbabilities,\hat\baseRates,\kernels^{\hat\kernelsParams})$ process can be assessed as follows.} Denote by \mikko{$\collection{T^e_n}{n=1,\ldots,n_e}$} the sequence of times at which an event of type $e$ occurred and by \mikko{$\collection{T^{ex}_n}{n=1,\ldots,n_{ex}}$} the sequence of times at which an event type $e$ occurred and after which the state was $x$, and set \mikko{$T^e_0=T^{ex}_0=0$} \mikko{for all} $e\in\eventSpace$ \mikko{and} $x\in\stateSpace$. \mikko{Introduce additionally} the event residuals $\residualOfEvent{e}{n}$ and total residuals $\residualOfEventState{e}{x}{n}$
\begin{align*}
	\residualOfEvent{e}{n} &:= \mikko{\int_{T^{e}_{n-1}}^{T^{e}_{n}}  \hatIntensityOfEvent{e}(t)dt,} \quad n=1,\ldots,\mikko{n_e}, \quad e\in\eventSpace, \\
	\residualOfEventState{e}{x}{n} &:= \mikko{\int_{T^{ex}_{n-1}}^{T^{ex}_{n}} \hatintensityOfEventState{e}{x}(t)dt,} \quad n=1,\ldots,\mikko{n_{ex}}, \quad e\in\eventSpace, \quad x\in\stateSpace\mikko{,}
\end{align*}
where \mikko{$\hat{\intensityOfEvents}$} and \mikko{$\hat{\intensityOfEventsStates}$} are computed by plugging \mikko{the estimated values $(\hat\transitionProbabilities,\hat\baseRates,\kernels^{\hat\kernelsParams})$} into \eqref{eq:event_intensity} and \eqref{eq:hybrid_intensity}, respectively.

It is a classical result that under the right \mikko{time change}, $\N$ and $\Nhybrid$ become \mikko{unit-rate} Poisson processes \citep{Meyer:1971}, with the event residuals and total residuals as their respective time increments\mikko{, which applies to the present setting as follows.}
\newtheorem{residuals}[algorithm]{Theorem}
\begin{residuals} \label{thm:residuals}
\begin{enumerate}[label=\textup{(\roman*)}]
\item \mikko{Suppose that $\collection{T^e_n}{n=1,\ldots,n_e}$ are generated by a $\numberOfEvents$-dimensional multivariate point process $\N$ with an $\internalHistoryCollection{\N, X}$-intensity $\intensityOfEvents$ satisfying \eqref{eq:event_intensity}, where $X$ is a given state process, $\baseRates=\hat\baseRates$ and $\kernels=\kernels^{\hat\kernelsParams}$. Then the event residuals $\collection{\residualOfEvent{e}{n} }{n=1,\ldots,\mikko{n_e}}$ for any $e \in \eventSpace$ are i.i.d.\ and follow the $\mathrm{Exp}(1)$ distribution.}
\item \mikko{Suppose that $\collection{T^{ex}_n}{n=1,\ldots,n_{ex}}$ are generated by an $\mathrm{sdHawkes}(\hat\transitionProbabilities,\hat\baseRates,\kernels^{\hat\kernelsParams})$ process $(\N, X)$. If moreover $\NhybridOfEventState{e}{x}(t)\rightarrow\infty$ as $t\rightarrow\infty$ with probability one for any $e\in\eventSpace$ and $x\in\stateSpace$, then the total residuals $\collection{\residualOfEventState{e}{x}{n}}{n=1,\ldots,\mikko{n_{ex}}}$ for any $e\in\eventSpace$ and $x\in\stateSpace$ are i.i.d.\ and follow the $\mathrm{Exp}(1)$ distribution.}
\end{enumerate}
\end{residuals}

Consequently, \mikko{the} goodness-of-fit \mikko{of the process} can be assessed by comparing the \mikko{empirical} distribution of \mikko{each sequence among $\collection{\residualOfEvent{e}{n} }{n=1,\ldots,\mikko{n_e}}$, $e\in \eventSpace$ and $\collection{\residualOfEventState{e}{x}{n}}{n=1,\ldots,\mikko{n_{ex}}}$, $e\in\eventSpace$, $x\in\stateSpace$ to the $\mathrm{Exp}(1)$ distribution,} \maxime{which can by achieved by inspecting a Q--Q plot (see Subsection \ref{subsec:goodness_of_fit}). Goodness-of-fit diagnostics can also check whether the residuals are mutually independent, which can be assessed by examining their correlogram. One can of course go beyond these visual assessments by applying formal statistical tests \citep{Bowsher2007}.}

\subsection{The special case of exponential kernels} \label{subsec:exp_kernels}

The direct computation of the term $\sum_{n=1}^{N} \ln \intensityOfEvent{e_n}(t_n)$ in the \mikko{log likelihood function} \eqref{eq:log-likelihood} involves a double sum \mikko{requiring} $\mathcal{O}(N^2)$ \mikko{operations to evaluate, which may render ML estimation numerically infeasible when the sample is large.} However, for \mikko{ordinary} Hawkes processes ($\numberOfStates=1$) with exponential kernels, it is known that this computation can be achieved in $\mathcal{O}(N \numberOfEvents^2)$ \mikko{operations} \citep{Ozaki1979a, Ogata:1981aa}. \mikko{Fortunately,} this property carries over to state-dependent Hawkes processes with \mikko{kernel} $\kernels$ \mikko{having exponential form}
\begin{equation} \label{eq:kernels_exp}
	\kernelAtoB{e'}{e}(t,x) = \impactCoeff{e'}{x}{e}\exp(-\decayCoeff{e'}{x}{e}t),\quad t>0, \quad e',e\in\eventSpace, \quad x\in\stateSpace,
\end{equation}
where the \mikko{\emph{impact coefficients}} $\impactCoeffs:=\collection{\impactCoeff{e'}{x}{e}}{}$ and \mikko{\emph{decay coefficients}} $\decayCoeffs:=\collection{\decayCoeff{e'}{x}{e}}{}$ are non-negative. \mikko{The reduction in computational cost becomes apparent from the derivation in Appendix \ref{subsec:formulae_liklihood}. In particular, \mikkoagain{the term $\sum_{n=1}^{N} \ln \intensityOfEvent{e_n}(t_n)$ can be expressed in terms of sums that satisfy a convenient recursive relationship,} reducing the computational cost} from $\mathcal{O}(N^2)$ to $\mathcal{O}(N)$ \mikko{operations}. A similar remark holds for the computation of the gradient and the residuals. Moreover, the $\mathcal{O}(N)$ complexity extends also to any kernel whose components are linear combinations of functions of the form \eqref{eq:kernels_exp}.

\section{Application to high-frequency \mikko{limit order book} data} \label{sec:application}

\subsection{Limit order book mechanism}

\mikko{We first briefly review the mechanics} of a \mikko{limit order book and recall} the definition\mikko{s} of some key market quantities\mikko{, following} \citet{gould2013surveylimit}.

\mikko{In} order-driven markets, market participants submit orders to buy or sell an asset (e.g., a stock) at the price of their choice. \mikko{Formally, an \emph{order}} is defined by \mikko{its \emph{submission time}, \emph{direction} (buy or sell), \emph{price} $p$ and \emph{size} $q$.} The quantities $p$ and $q$ must \mikko{typically} be multiples of the tick size and lot size, respectively, which are fixed by the exchange \mikko{or the market regulator}. When a buy order is submitted, if there are \mikko{unfilled} sell orders with prices $p' < p$, the buy order is matched at the smallest price $p'$ with the oldest order(s)\mikko{, following the \emph{price--time priority} rule}. \mikko{Otherwise} the order becomes \mikko{\emph{active}}, that is, it enters the queue of \mikko{unfilled} orders with price $p$. Such orders are called \textit{limit orders}. The collection of \mikko{current} limit orders is called the \textit{limit order book} \mikko{(LOB)} and can thus be understood as a snapshot of the expressed supply and demand \mikko{or \emph{visible liquidity}}. Orders that result in an instant match with pre-existing limit orders are called \textit{market orders}. \mikko{Orders may also be \emph{cancelled}, that is, an unfilled or partially filled order is withdrawn from the LOB.}

The highest (respectively lowest) price among buy (respectively sell) active limit orders is called the \textit{bid} (respectively \textit{ask}) price. For instance, the bid price is the best price at which one can instantly sell by sending a market order. The difference between the ask price and bid price is called the \mikko{\emph{bid--ask spread}}. The orders and cancellations with prices that \mikko{are equal to or more aggressive than} the \mikko{bid and ask} prices at their submission time form the \textit{level-I order flow.}

\subsection{Data}

\begin{figure}[!t]
\begin{subfigure}[t]{.5\textwidth}
  \centering
  \includegraphics[width=.95\linewidth]{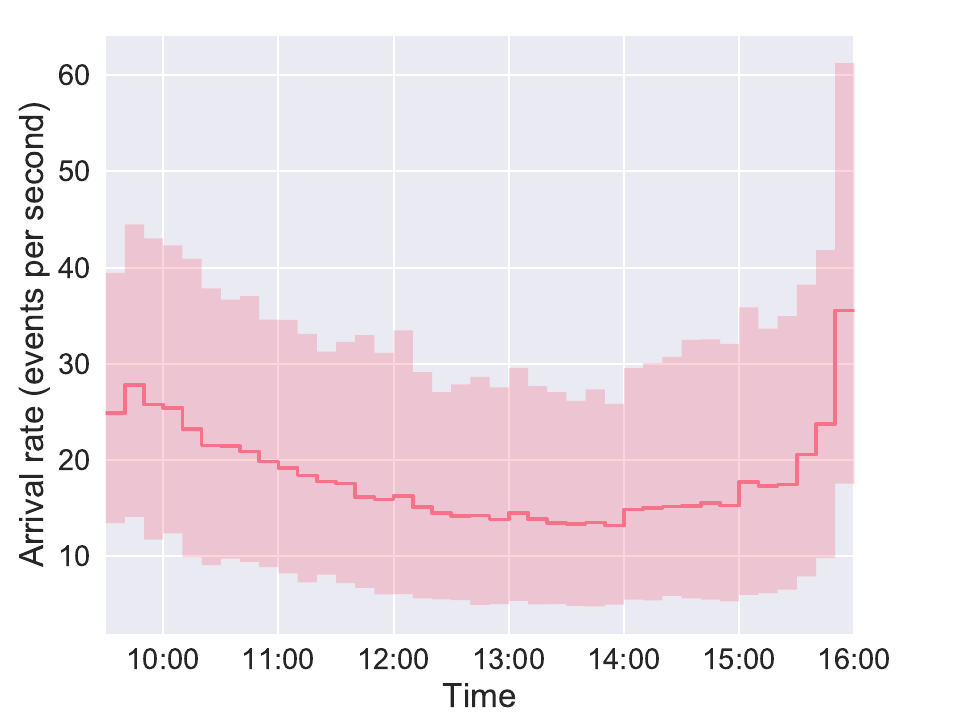}
  \caption{Arrival rate of level-I orders.}
  \label{fig:u_shape}
\end{subfigure}
\begin{subfigure}[t]{.5\textwidth}
  \centering
  \includegraphics[width=.95\linewidth]{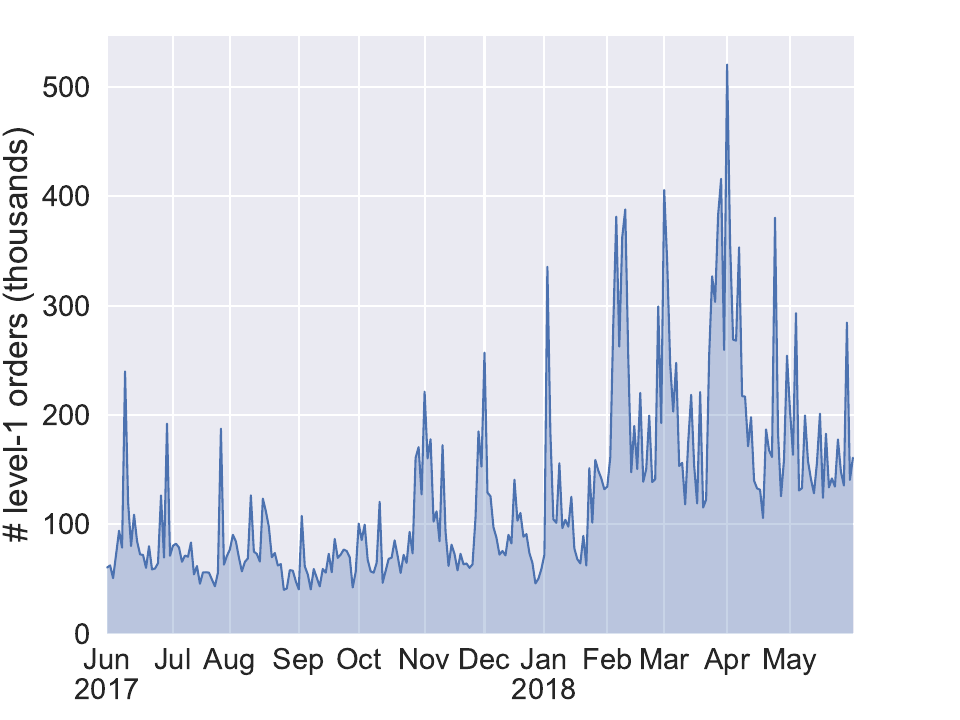}
  \caption{Number of level-I orders.}
  \label{fig:n_orders}
\end{subfigure}
\\
\begin{subfigure}[t]{.5\textwidth}
  \centering
  \includegraphics[width=.95\linewidth]{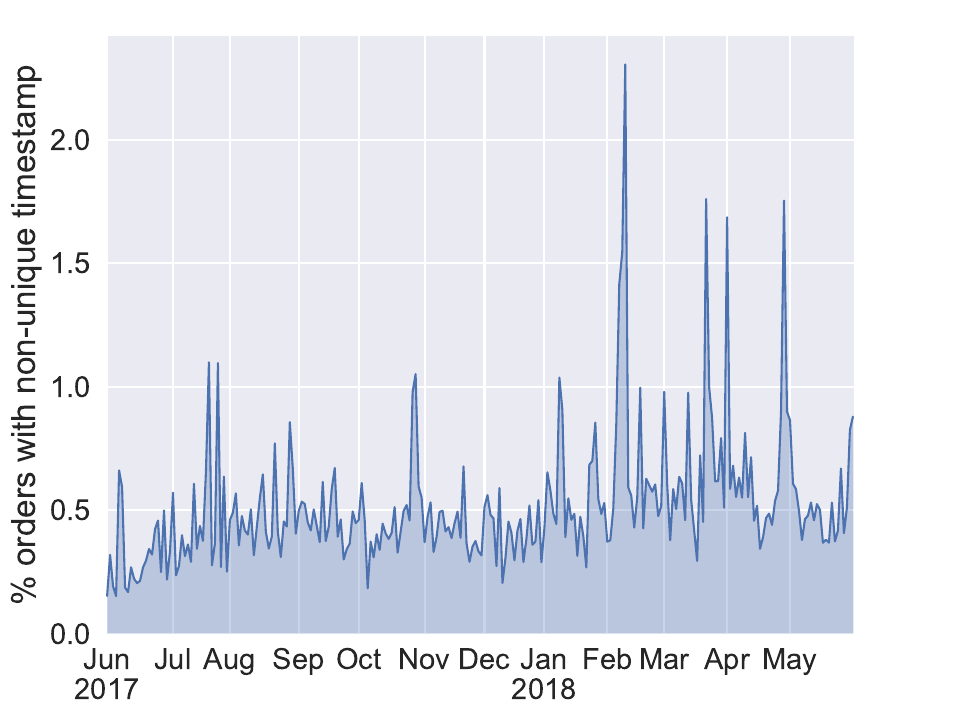}
  \caption{Fraction of level-I orders with non-unique timestamp.}
  \label{fig:same_time}
\end{subfigure}
\begin{subfigure}[t]{.5\textwidth}
  \centering
  \includegraphics[width=.95\linewidth]{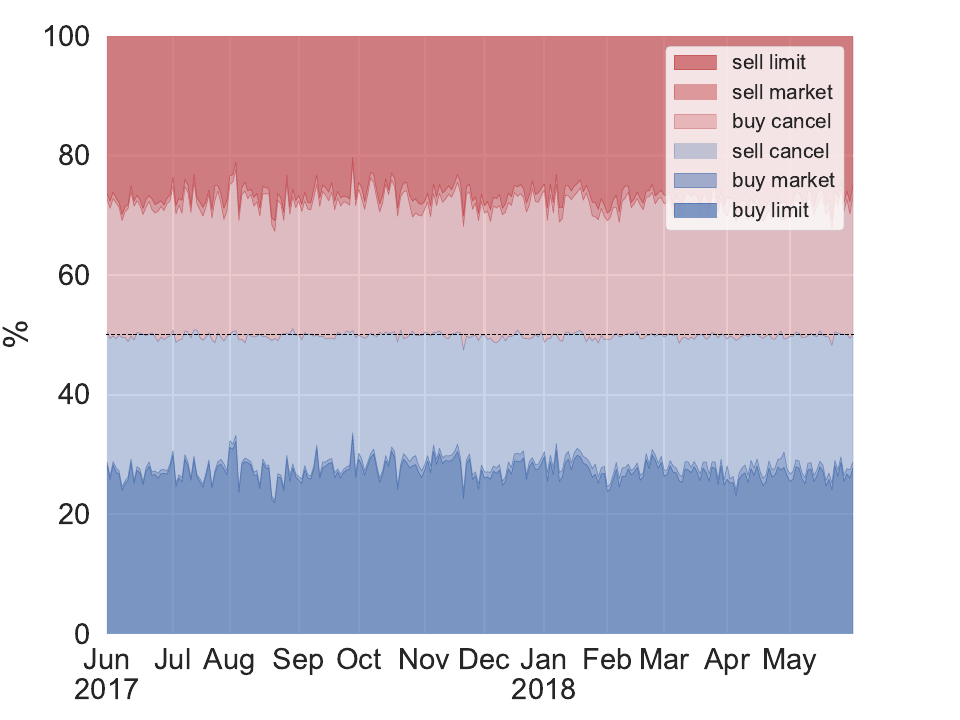}
  \caption{Distribution of level-I order types.}
  \label{fig:order_distribution}
\end{subfigure}
\caption{Descriptive statistics of level-I order flow \mikko{of INTC}. Except for Figure \ref{fig:u_shape}, only the data \mikko{between 12:00 and 14:30 are} used. In Figure \ref{fig:u_shape}, the sample mean of the arrival rate of level-I orders is computed over 10-minute bins. The translucent area represents the range of the arrival rate across the 250 trading days, excluding the bottom and top 5\% values.}
\label{fig:descriptive_stats}
\end{figure}

\mikko{We analyse tick-by-tick level-I LOB data on the stock of Intel Corporation (INTC), traded on the Nasdaq Stock Market, from 1 June 2017 to 31 May 2018. We discard the data on 3 June 2017 (the day before Independence Day) and 24 November 2017 (Black Friday, i.e., the day after Thanksgiving) because of early market close on those days at 13:00, leaving 250 full trading days in the data set. Besides INTC, we have also studied in the same manner the stocks of Advanced Micro Devices Inc.\ (AMD), Micron Technology Inc.\ (MU), Snap Inc.\ (SNAP) and Twitter Inc.\ (TWTR) from January to April 2018, but due to space constraints we only report our results on INTC as they are largely representative of the findings on the other stocks. (Full results are available from the authors upon request.)}

Whilst the primary listing of INTC is the Nasdaq Stock Market, it can be traded on several stock exchanges and \emph{alternative trading systems} (ATS) across the increasingly fragmented US equity market system \citep{O_Hara_2011}. In particular, there is no single LOB that would in real time aggregate the available liquidity on all trading venues --- the Nasdaq LOB of INTC represents only a part of the visible liquidity. The LOBs of different exchanges cannot, however, diverge significantly, at least in terms of prices, due to the arbitrage opportunities that would ensue and the \emph{National Best Bid and Offer} (NBBO) rule that mandates brokers to execute trades at the best available prices in the market system. Moreover, we deem the Nasdaq LOB of INTC to be representative of the state of the market since it had the largest market share of INTC among all exchanges during the observation period \mikkoagain{ --- according to Nasdaq trading volume statistics (\url{https://www.nasdaqtrader.com/trader.aspx?ID=marketsharedaily}), around 30\% of the total INTC market volume was traded on Nasdaq, which is a typical figure for a stock with Nasdaq as primary listing.}


\mikko{In the US equity markets, the tick size is fixed to \$0.01 by Rule 612 of \emph{Regulation National Market System} (Reg NMS), with the exception of stocks priced below \$1.00 per share. This means that for a stock with a low price per share, the uniform tick size is relatively large compared to the price. Such stocks are dubbed \emph{large-tick stocks}. Whilst the exact criterion is subjective, the price of INTC during the sample period was until January 2018 below the threshold of \$50 per share used by \citet{Bonart2017} to distinguish large-tick stocks and below \$60 during the entire period.}
 A characteristic \mikko{feature} of large-tick stocks is that \mikko{their bid--ask} spread is most of the time equal to one tick, which \mikko{is confirmed for INTC in} Figure \ref{fig:events_spread_distribution}. \mikko{Another important feature of large-tick stocks is that most of their liquidity and trading activity is concentrated on the bid and ask levels, which is also our rationale for focusing on level-I data on INTC and eschewing deeper levels of the LOB.}

\mikko{In the modelling part of our analysis, we} focus \mikko{exclusively on trading activity} \mikko{between 12:00 and 14:30 for the following reason}. \mikko{It is well-known, and also exemplified in Figure \ref{fig:u_shape}, that the intensity of trading} activity is not constant throughout the day but \mikko{follows on average} a U-shape\mikko{d curve}. \mikko{Because of this diurnal pattern,} \mikko{imposing a} constant base rate \mikko{vector} $\baseRates$ over the entire trading day \mikko{might} result in overestimation of the self- and cross-excitation \mikko{effects} \citep{Rambaldi2015, Omi2017}. \mikko{The intensity of trading activity is relatively constant over the chosen intraday period, whilst the choice still leaves an ample amount of observations for model estimation} (at least 50,000 level-I orders \mikko{on} each day, see Figure \ref{fig:n_orders}).

\mikko{The data was supplied by \mikko{LOBSTER} (\url{https://lobsterdata.com}) in a form where the LOB at the time of each event has already been reconstructed. The timestamp of each event is} recorded with nanosecond precision. As a result, more than 99\% of level-I orders have a unique timestamp (Figure \ref{fig:same_time}). This contrasts with \mikko{some earlier} studies \citep{Bowsher2007, Large:2007aa:MeasureResiliency} \mikko{where} lower timestamp resolutions (e.g., one second) \mikko{were used, leading to a significant amount of \emph{tied} timestamps, shared by multiple events.} As Hawkes processes capture the \mikko{\emph{Granger causality}} between different event types \citep{Eichler2017,EmbrechtsandKirchner2018}, being able to accurately \mikko{establish the order of events}, even at the \mikko{shortest} timescales, is \mikko{essential for accurate estimation of Hawkes processes}.
\maxime{It is worth pointing out that in LOBSTER data, a market order that is matched with, say, $n$ multiple limit orders is recorded as $n$ individual market orders sharing the same timestamp. However, in our analysis, including Figure \ref{fig:same_time}, we address this artefact by aggregating market orders with a tied timestamp and counting them as a single level-I event.}

\subsection{Model specification}

\begin{table}[!t]
\centering
\mikko{\begin{tabular}{lcc}
\toprule
                    & $\mathrm{Model}_{\mathrm{S}}$                                                       & $\mathrm{Model}_{\mathrm{QI}}$                                                                \\ \midrule
Event types        & \multicolumn{2}{c}{\event{ask} and \event{bid} events, $\numberOfEvents=2$, $\eventSpace=\{\mbox{\event{ask}, \event{bid}}\}$}                                           \\
State process       & \begin{tabular}[t]{@{}c@{}}Bid--ask spread, $\numberOfStates=2$\\ $\stateSpace=\{\mbox{\state{1}}, \mbox{\state{2+}}\}$\end{tabular} & \begin{tabular}[t]{@{}c@{}}Queue imbalance, $\numberOfStates=5$\\ $\stateSpace = \{ \mbox{\state{sell++}}, \mbox{\state{sell+}}, \mbox{\state{neutral}}, \mbox{\state{buy+}}, \mbox{\state{buy++}}\}$\end{tabular} \\
Kernel $\kernels$             & \multicolumn{2}{c}{Exponential, of the form \eqref{eq:kernels_exp}}                                                                                                        \\
No.\ of parameters & 26                                                            & 92                                                                     \\
\bottomrule
\end{tabular}}
\caption{Summary of \mikko{$\mathrm{Model}_{\mathrm{S}}$ and $\mathrm{Model}_{\mathrm{QI}}$}.}
\label{table:models}
\end{table}

\mikko{We work with models, based on state-dependent Hawkes processes, that distinguish t}wo event types, \mikko{denoted} \event{ask} and \event{bid} (\mikko{that is,} $\eventSpace=\{\mbox{\event{ask}, \event{bid}}\}$, $\numberOfEvents=2$). \mikko{In our analysis, \event{bid} events consist of} buy market orders, level-I buy limit orders and level-I sell cancellations\mikko{, whilst \event{ask} events inversely consist of sell market orders, level-I sell limit orders and level-I buy cancellations.} Consequently, $\NofEvent{\textsf{\event{bid}}} - \NofEvent{\textsf{\event{ask}}}$ can be interpreted as a proxy of the order flow imbalance, which was shown by \citet{cont2013price} to be the main driver of price changes. \mikko{More concretely, \event{bid} events tend to push the price up and \event{ask} events down.} From Figure \ref{fig:order_distribution}, we see that the distribution between these two event types is very \mikko{balanced}, with market orders accounting for less than 5\% of the level-I activity. \mikko{As discussed above,} only the level-I events are modelled \mikko{and} events occurring \mikko{at} deeper \mikko{levels of} the \mikko{LOB} discarded.
One could of course increase $\numberOfEvents$ to \mikko{have a} more granular \mikko{classification of} event types. Still, the focus of this paper is on \mikko{the} modelling \mikko{of state dependence} and the choice $\eventSpace=\{\mbox{\state{ask}, \state{bid}}\}$ \mikko{already leads to} interesting results whilst keeping the dimensionality low, which makes the results \mikko{easier to visualise}.

\mikko{As the state variable} we consider the \mikko{bid--ask} spread and the queue imbalance\mikko{, giving rise to two models dubbed $\mathrm{Model}_{\mathrm{S}}$ and $\mathrm{Model}_{\mathrm{QI}}$}, respectively. In \mikko{$\mathrm{Model}_{\mathrm{S}}$}, we set $\stateSpace=\{\mbox{\state{1}}, \mbox{\state{2+}}\}$ ($\numberOfStates=2$), \mikko{where the states correspond to the bid--ask spread being one tick ($X(t)=\mbox{\state{1}}$) and two ticks or more ($X(t)=\mbox{\state{2+}}$). Increasing the number of states beyond $\numberOfStates=2$ in this setting would not be} practically relevant since the \mikko{bid--ask} spread is very rarely strictly \mikko{wider} than two ticks. The queue imbalance\mikko{, used in $\mathrm{Model}_{\mathrm{QI}}$,} is \mikko{nowadays} recognised as a popular trading signal \mikko{with predictive power on} the direction of the next price move \citep{Cartea:Donelly:2018}. Denoting the total size of limit orders sitting at the ask price by \mikko{$Q_{\mathrm{ask}}(t)$, and defining $Q_{\mathrm{bid}}(t)$ analogously,} the queue imbalance \mikko{can be expressed} as
\begin{equation*}
	\mathrm{QI}(t):=\frac{Q_{\mathrm{bid}}(t)-Q_{\mathrm{ask}}(t)}{Q_{\mathrm{bid}}(t)+Q_{\mathrm{ask}}(t)} \in [-1, 1].
\end{equation*}
\mikkoagain{(The denominator $Q_{\mathrm{bid}}(t)+Q_{\mathrm{ask}}(t)$ equals zero if and only if the LOB is empty at time $t$. In this case, it would be natural to define $\mathrm{QI}(t):=0$, but this never occurs in our data set.)}
For example, \mikko{the condition} $\mathrm{QI}(t)>0$ \mikko{signals} buy pressure and tends to be followed by an upwards price move. As in \citet{Cartea:Donelly:2018}, we split the interval $[-1, 1]$ into $\numberOfStates=5$ bins of equal width which we label as follows:
\begin{equation}\label{eq:bins}
	\stateSpace = \{ \underbrace{\mbox{\state{sell++}}}_{[-1, -0.6)}, \underbrace{\mbox{\state{sell+}}}_{[-0.6, -0.2)}, \underbrace{\vphantom{\mbox{+}}\mbox{\state{neutral}}}_{[-0.2, 0.2)}, \underbrace{\mbox{\state{buy+}}}_{[0.2, 0.6)}, \underbrace{\mbox{\state{buy++}}}_{[0.6, 1]} \}\mikko{,}
\end{equation}
\mikko{whereby in $\mathrm{Model}_{\mathrm{QI}}$ the state variable $X(t)$ indicates bin where $\mathrm{QI}(t)$ is located.}

\mikko{Finally, given the large number of observations we are dealing with (see Figure \ref{fig:n_orders}), we use the exponential specification \eqref{eq:kernels_exp} of the kernel $\kernels$ in both models, as it leads to a significant reduction of computational cost in estimation, as discussed in Subsection \ref{subsec:exp_kernels}. The full specifications of the models are summarised in Table~\ref{table:models}.}


\subsection{\mikko{Visualising estimated excitation effects}}

\mikko{We use the following approach to present our estimation results on self- and cross-excitation effects.} For each triple $(e,e',x)\in\eventSpace^2 \times \stateSpace$, \mikko{ML} estimation \mikko{produces an estimated excitation profile $t\mapsto\kernelAtoB{e'}{e}(t, x)$, parameterised by the estimated} impact coefficient $\impactCoeffHat{e'}{x}{e}$ and decay coefficient $\decayCoeffHat{e'}{x}{e}$. However, \mikko{instead of reporting} $\impactCoeffHat{e'}{x}{e}$ and $\decayCoeffHat{e'}{x}{e}$, \mikko{we visualise the excitation profile by plotting the truncated $L^1$-norm}
\begin{equation*}
	t\mapsto \|\kernelAtoBHat{e'}{e}(\cdot, x)\|_{1,t}:=  \int_0^t\kernelAtoBHat{e'}{e}(s,x)ds,
\end{equation*}
\mikko{from which the magnitude and the effective timescale of the excitation effect is easier to gauge than from the numerical values of $\impactCoeffHat{e'}{x}{e}$ and $\decayCoeffHat{e'}{x}{e}$.} \mikko{In fact, taking guidance} from the cluster representation of \mikko{ordinary} Hawkes processes \citep{Hawkes:1974aaClusterRepresentation}, \mikko{we can conveniently} interpret $\|\kernelAtoBHat{e'}{e}(\cdot, x)\|_{1,t}$ as the average number of events of type $e$ that have been directly triggered by an event of type $e'$ \mikko{in state $x$} within $t$ seconds of its occurrence. 
\mikko{Further, we note} that the \mikko{full $L^1$-}norm is given by 
\begin{equation*}
\mikko{\|\kernelAtoBHat{e'}{e}(\cdot, x)\|_{1,\infty} = \lim_{t\rightarrow \infty} \|\kernelAtoBHat{e'}{e}(\cdot, x)\|_{1,t} = \frac{\impactCoeffHat{e'}{x}{e}}{\decayCoeffHat{e'}{x}{e}}.}
\end{equation*}

\subsection{\mikko{Estimation results for $\mathrm{Model}_{\mathrm{S}}$}} \label{subsec:results_1}
 
 \begin{figure}[!t]
\begin{subfigure}[t]{.4\textwidth}
  \centering
  \includegraphics[width=.48\linewidth]{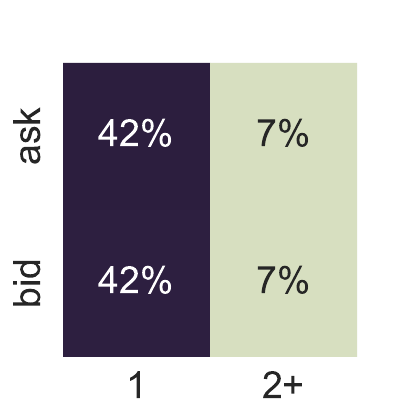}
  \caption{$\mathrm{Model}_{\mathrm{S}}$ (state process: bid--ask spread).}
  \label{fig:events_spread_distribution}
\end{subfigure}
\begin{subfigure}[t]{.6\textwidth}
  \centering
  \includegraphics[width=.8\linewidth]{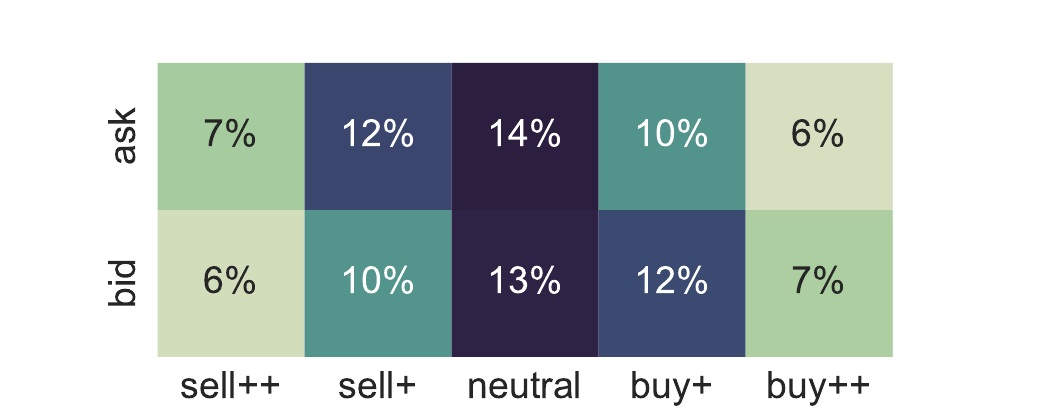}
  \caption{$\mathrm{Model}_{\mathrm{QI}}$ (state process: queue imbalance).}
  \label{fig:events_imbalance_distribution}
\end{subfigure}
\caption{Joint distribution of events and states \mikko{for INTC, depicting the empirical} distribution of \mikko{the marks} $(E_n, X_n)$ for the two considered state processes.}
\label{fig:events_states_distribution}
\end{figure}

\mikko{We estimate} the model parameters $(\hat{\transitionProbabilities}^{(i)}, \hat\baseRates^{(i)}, \hat\impactCoeffs^{(i)}, \hat\decayCoeffs^{(i)})$ \mikko{for each} trading day $i$ \mikko{in the sample by ML estimation,} as explained in Section \ref{sec:MLE}. Practical details \mikko{on the numerical solution of the} underlying optimisation problem can be found in the \mikko{Appendix \ref{app:numerical_opt}.}

\mikko{The estimated transition distribution $\hat{\transitionProbabilities}$ of $\mathrm{Model}_{\mathrm{S}}$, obtained by averaging over the daily estimates $\hat{\transitionProbabilities}^{(i)}$, is presented in Figure \ref{fig:probabilities_spread}.} The state process \mikko{describing the bid--ask spread} exhibits \mikko{persistent} behaviour, in the sense that the probability of remaining in the \mikko{current} state is very high. We also observe \mikko{higher likelihood} of moving from state \state{2+} to \state{1} than \mikko{vice versa, which is consistent with one-tick bid--ask spread being the equilibrium state for a large-tick stock like INTC}. \mikko{We also find that the} transition probabilities are not sensitive to the event type\mikko{, which is natural since the bid--ask spread not is expected to be influenced by the direction of orders, ceteris paribus.}

\mikko{The estimation results on the excitation kernel, given in Figure \ref{fig:kernels_spread}, indicate that self-excitation effects surpass cross-excitation effects in both states. Their magnitude and timescales, however, vary between the two states. The effective timescales of these effects range from 0.1 to 100 milliseconds, which is in agreement with the predominantly algorithmic origin and multiscale nature of trading in modern electronic markets.}

\begin{figure}[!t]
\begin{subfigure}{1\textwidth}
  \centering
  \includegraphics[width=.4\linewidth]{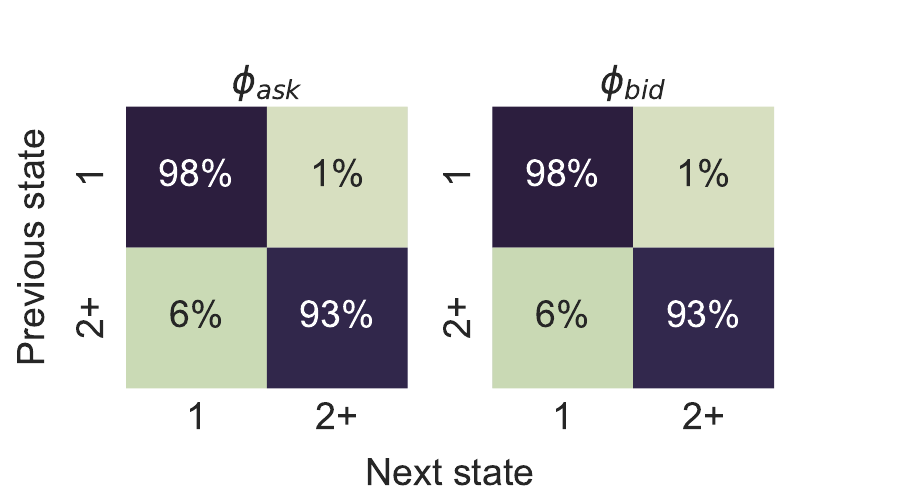}
  \caption{Transition probabilities of the bid--ask spread ($\mathrm{Model}_{\mathrm{S}}$).}
  \label{fig:probabilities_spread}
\end{subfigure}%
\\
\begin{subfigure}{1\textwidth}
  \centering
  \includegraphics[width=.8\linewidth]{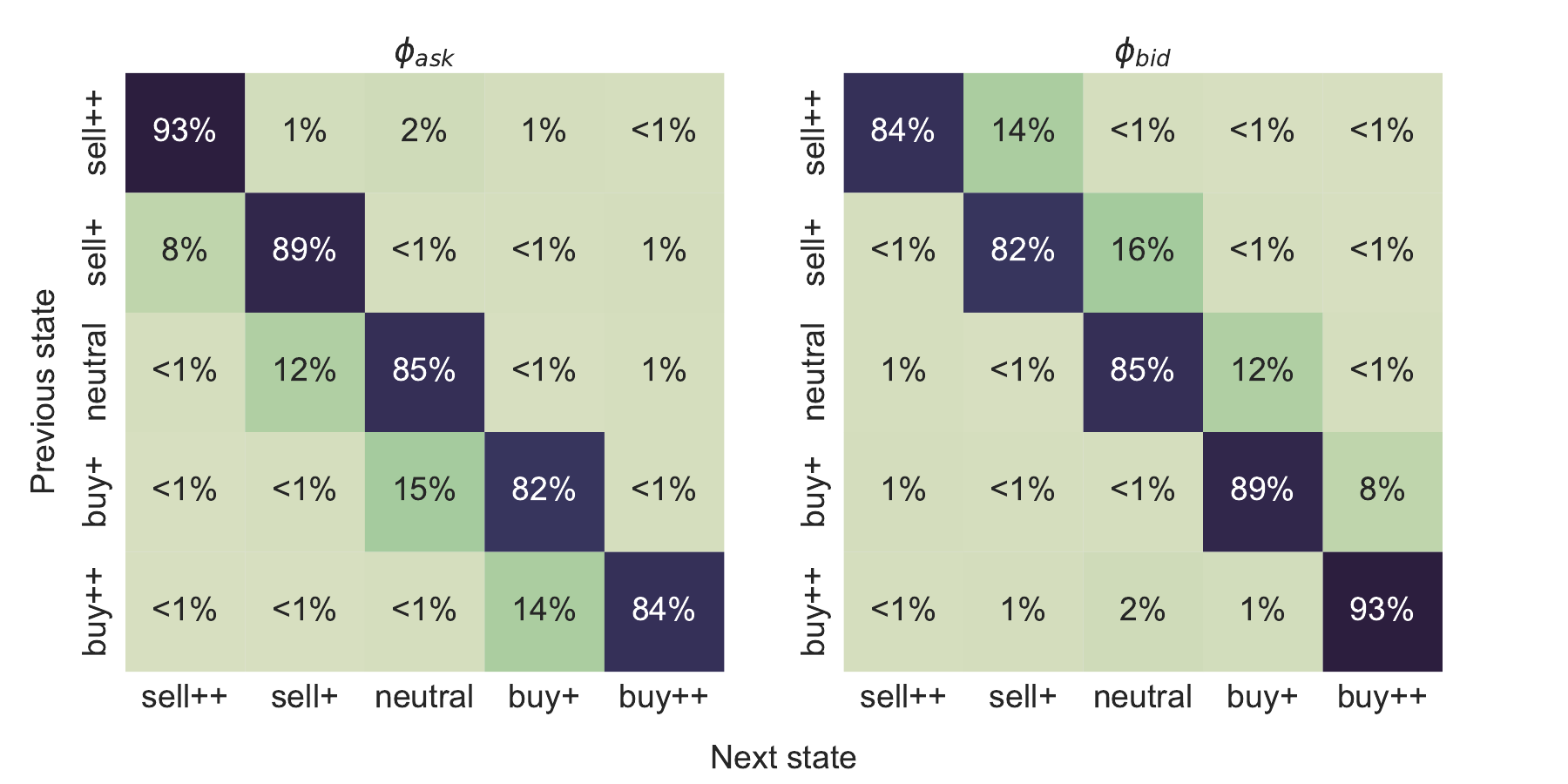}
  \caption{Transition probabilities of the queue imbalance ($\mathrm{Model}_{\mathrm{QI}}$).}
  \label{fig:probabilities_imbalance}
\end{subfigure}
\caption{Estimated transition distributions $\hat \transitionProbabilities$ of $\mathrm{Model}_{\mathrm{S}}$ and $\mathrm{Model}_{\mathrm{QI}}$. We report the average of $\hat\transitionProbabilities^{(i)}$ across the 250 trading days. \mikko{(Daily estimates vary little from these averaged values.)}}
\label{fig:transition_probabilities}
\end{figure}

\begin{figure}[p]
\begin{subfigure}{1\textwidth}
  \centering
  \includegraphics[width=.78\linewidth]{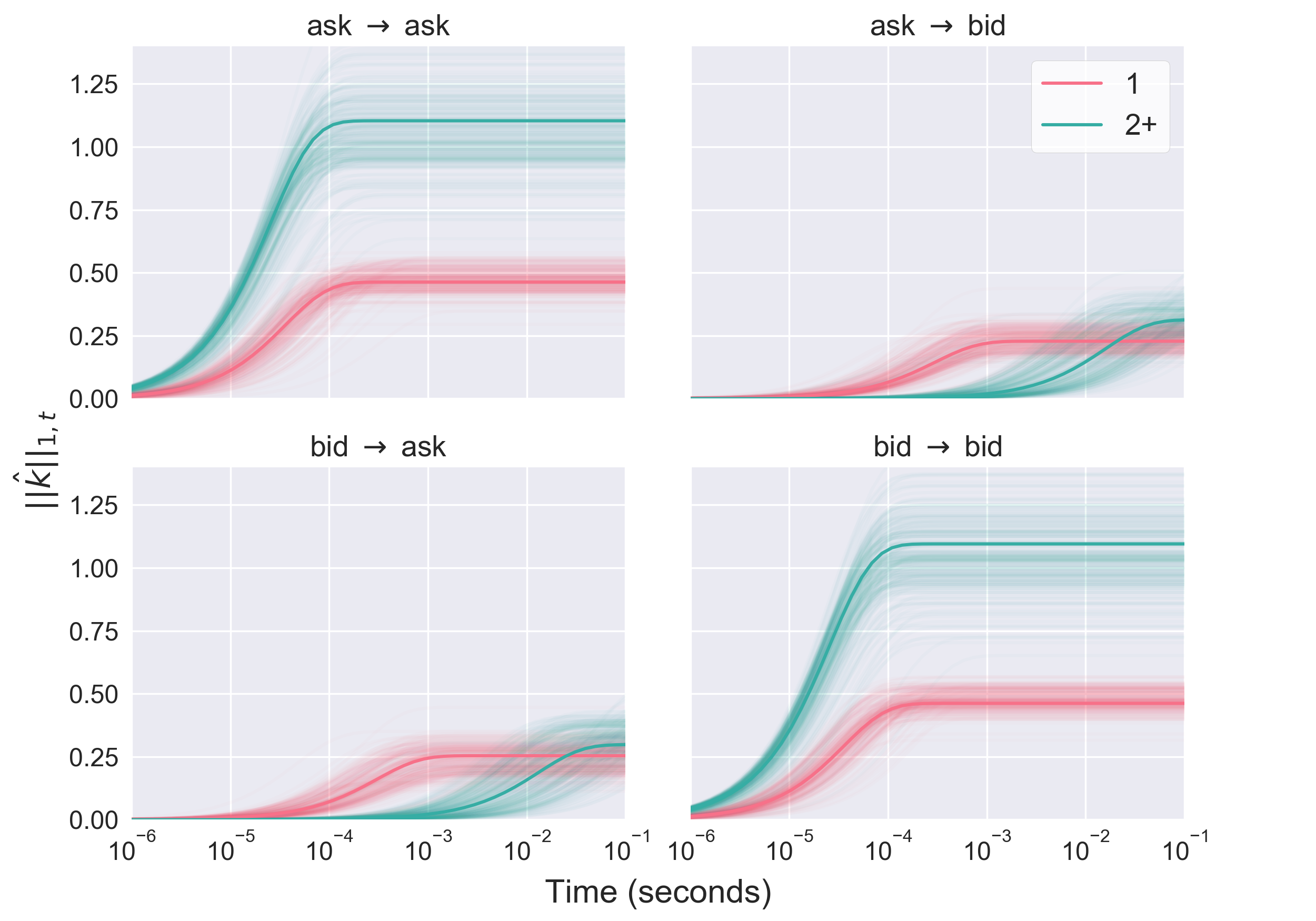}
  \caption{$\mathrm{Model}_{\mathrm{S}}$ (state variable: bid--ask spread).}
  \label{fig:kernels_spread}
\end{subfigure}%
\\
\begin{subfigure}{1\textwidth}
  \centering
  \includegraphics[width=.78\linewidth]{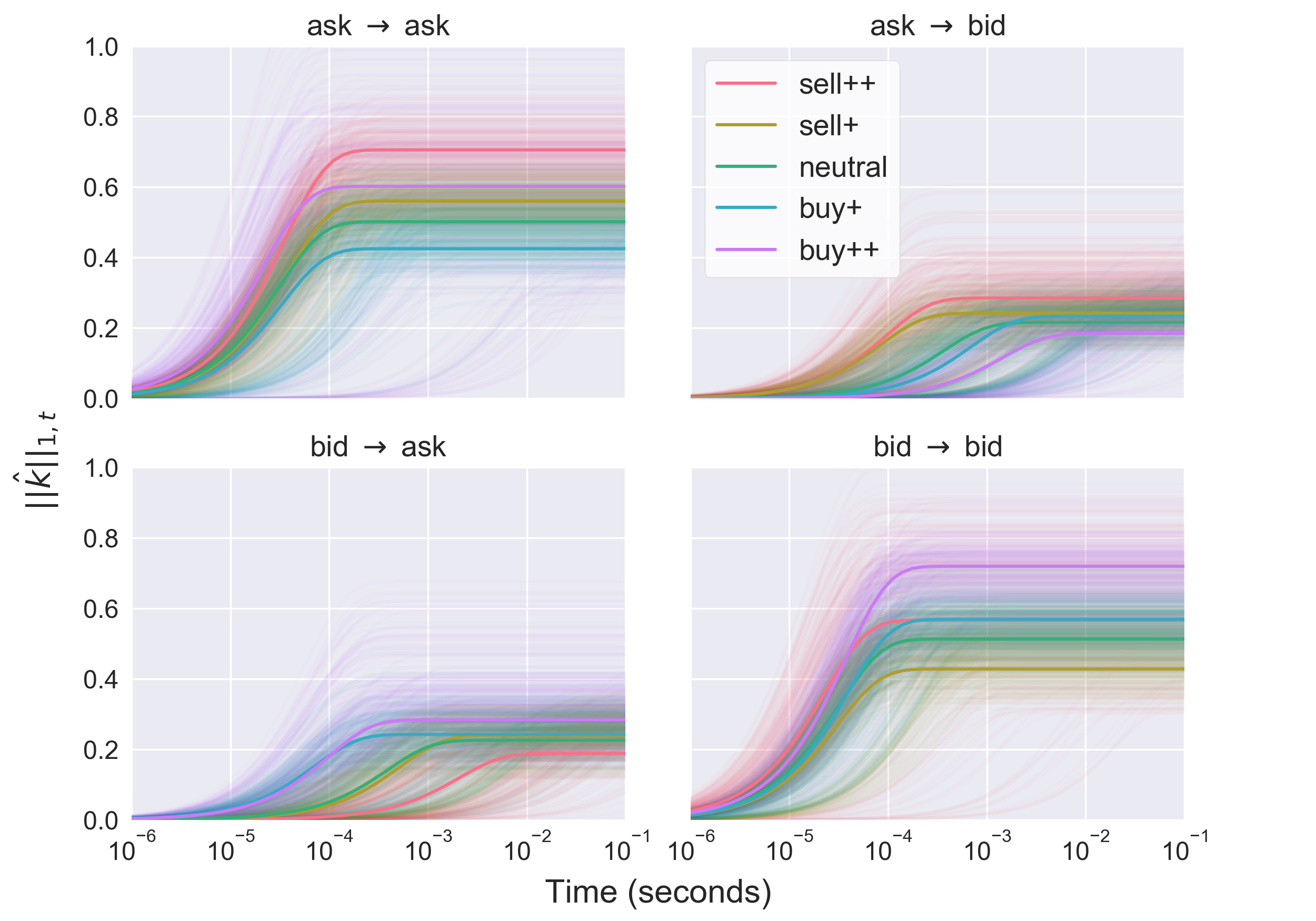}
  \caption{$\mathrm{Model}_{\mathrm{QI}}$ (state variable: queue imbalance).}
  \label{fig:kernels_imbalance}
\end{subfigure}
\caption{\mikko{The estimated kernel $\hat \kernels$ under $\mathrm{Model}_{\mathrm{S}}$ and $\mathrm{Model}_{\mathrm{QI}}$.} Each \mikko{panel} describes \mikko{self-} or cross-excitation as indicated by \mikko{its title}, whilst each colour corresponds to a different state. For example, in Figure \ref{fig:kernels_spread}, the red curves in the second \mikko{panel} represent the estimates $\kernelAtoBHat{e'}{e}^{(i)}(\cdot, x)$ where $e'=\mbox{\event{ask}}$, $e=\mbox{\event{bid}}$ and $x=\mbox{\state{1}}$. All \mikko{daily} estimates are superposed with one translucent curve for each day. \mikko{An ``aggregate''} kernel \mikko{is} represented by a solid line\mikko{, computed} using the median of $\hat\impactCoeffs^{(i)}$ and $\hat\decayCoeffs^{(i)}$ across the 250 trading days.}
\label{fig:kernels}
\end{figure}

When the \mikko{bid--ask} spread increases \mikko{to \state{2+}}, the magnitude of the self-excitation effects doubles whilst their \mikko{timescale remains roughly the same}. \mikko{The timescale of cross-excitation, however, lengthens drastically, whilst their magnitude increases slightly. A plausible microstructural explanation for this pattern goes as follows. When the bid--ask spread is in state \state{2+}, a trader can submit an aggressive limit order inside the spread, gaining queue priority at the cost of} a less favourable price. A \event{bid} event can then be seen as a signal for an upwards price move, which may prompt limit orders from buyers seeking a favourable position in the new queue and cancellations of limit orders from sellers trying to avoid adverse selection. \mikko{Traders who submit} sell orders are, per contra, incentivised to wait and see if the \mikko{expected} price \mikko{increase actually materialises}.

\begin{figure}[!t]
\begin{subfigure}[t]{.5\textwidth}
  \centering
  \includegraphics[width=.95\linewidth]{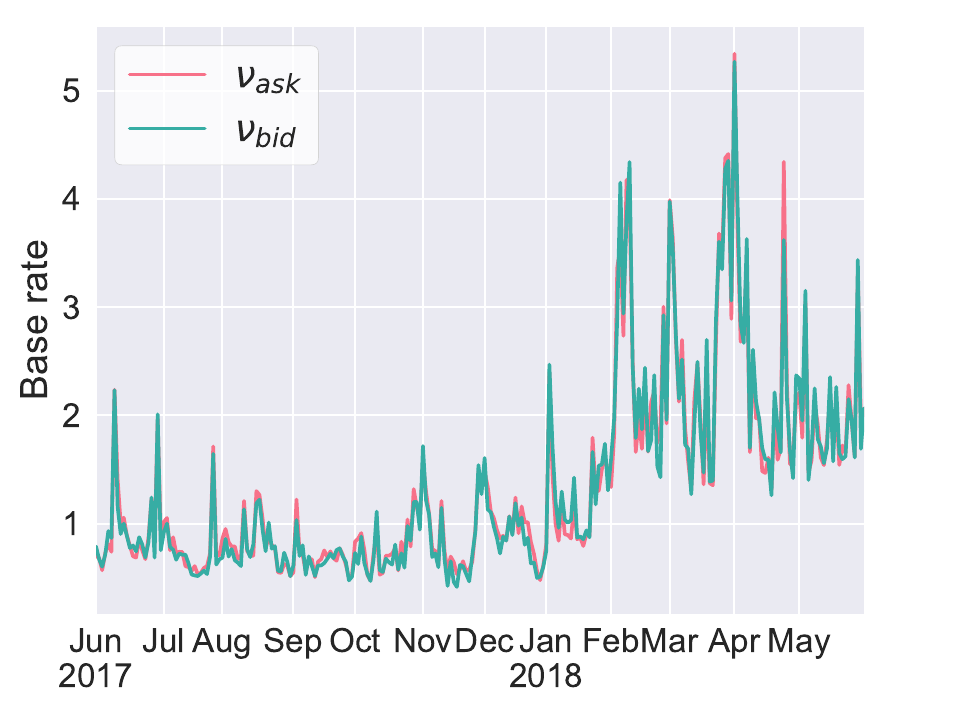}
  \caption{$\mathrm{Model}_{\mathrm{S}}$.}
  \label{fig:base_rates_S}
\end{subfigure}
\begin{subfigure}[t]{.5\textwidth}
  \centering
  \includegraphics[width=.95\linewidth]{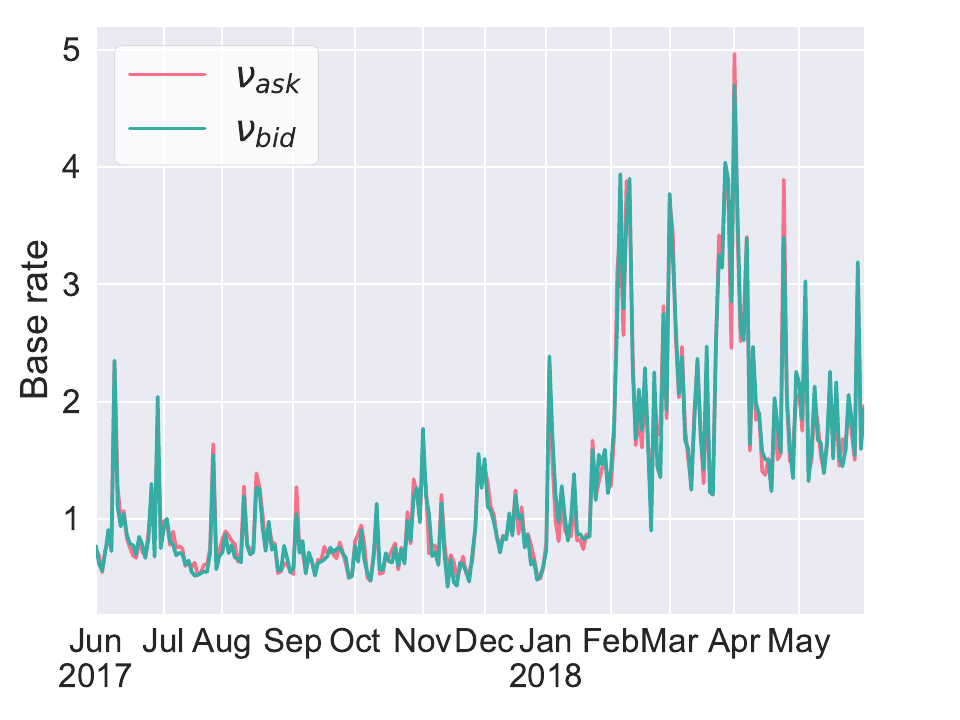}
  \caption{$\mathrm{Model}_{\mathrm{QI}}$.}
  \label{fig:base_rates_QI}
\end{subfigure}
\caption{\maxime{Estimated base rate vector $\hat\baseRates^{(i)}$ for $\mathrm{Model}_{\mathrm{S}}$ and $\mathrm{Model}_{\mathrm{QI}}$ over time (in number of events per second).}}
\label{fig:base_rates}
\end{figure}

\maxime{The evolution of the base rate vector $\hat\baseRates^{(i)}$ throughout the 250 days of data is displayed in Figure \ref{fig:base_rates}. We find a remarkable balance between buyers and sellers (i.e., $\baseRate{\mbox{\state{bid}}}\approx \baseRate{\mbox{\state{ask}}}$). We  also notice that the evolution of $\hat\baseRates^{(i)}$ mimics that of the total number of orders (Figure \ref{fig:n_orders}), which suggests that the day-to-day variation in market activity is mainly due to exogenous factors (cf.\ the analysis of endogeneity in Subsection \ref{sec:endo}).}

\subsection{\mikko{Estimation results for $\mathrm{Model}_{\mathrm{QI}}$}} \label{subsec:results_2}

\mikko{The estimated transition probabilities of $\mathrm{Model}_{\mathrm{QI}}$, presented in Figure \ref{fig:probabilities_imbalance}, convey a tendency to stick to the current state, similar to what is seen in $\mathrm{Model}_{\mathrm{S}}$. Here, however, this behaviour is more of an artefact --- each \event{ask} and \event{bid} event, by definition, changes the queue imbalance but not necessarily the state variable that is confined to the bins \eqref{eq:bins}.}

\mikko{In contrast to $\mathrm{Model}_{\mathrm{S}}$}, the estimated transition probabilities \mikko{now} depend on the event type \mikko{and we observe remarkable} mirror symmetry\mikko{, whereby $\hat\phi_{\textsf{\event{bid}}}$ equals, up to 1 percentage point, $\hat\phi_{\textsf{\event{ask}}}$ with the order of states reversed. This symmetry is natural, given the definition $\mathrm{QI}(t)$ --- a sell order always decreases} the queue imbalance unless it is submitted inside the \mikko{bid--ask} spread or \mikko{it} depletes the \mikko{current} bid queue\mikko{, whilst an analogous statement is true for buy orders.} As in \mikko{$\mathrm{Model}_{\mathrm{S}}$}, we find again that the probability of a state \mikko{transition is higher when it is towards the equilibrium state, here \state{neutral}, consistent with ideas about the \emph{resilience} of the LOB \citep{Large:2007aa:MeasureResiliency}.}

\mikko{Looking at the estimation results for the excitation kernel in Figure \ref{fig:kernels_imbalance}, we observe that, like in $\mathrm{Model}_{\mathrm{S}}$, self-excitation surpass cross-excitation, with the magnitude of the former and the timescale of the latter being manifestly sensitive to the current state. 
The mirror symmetry seen above in the context of the transition probabilities holds here as well, 
whereby it suffices to only speak about the results for \event{ask} events, whilst analogous conclusions can be drawn on \event{bid} events.}

\begin{figure}[t]
  \centering
  \includegraphics[width=.78\linewidth]{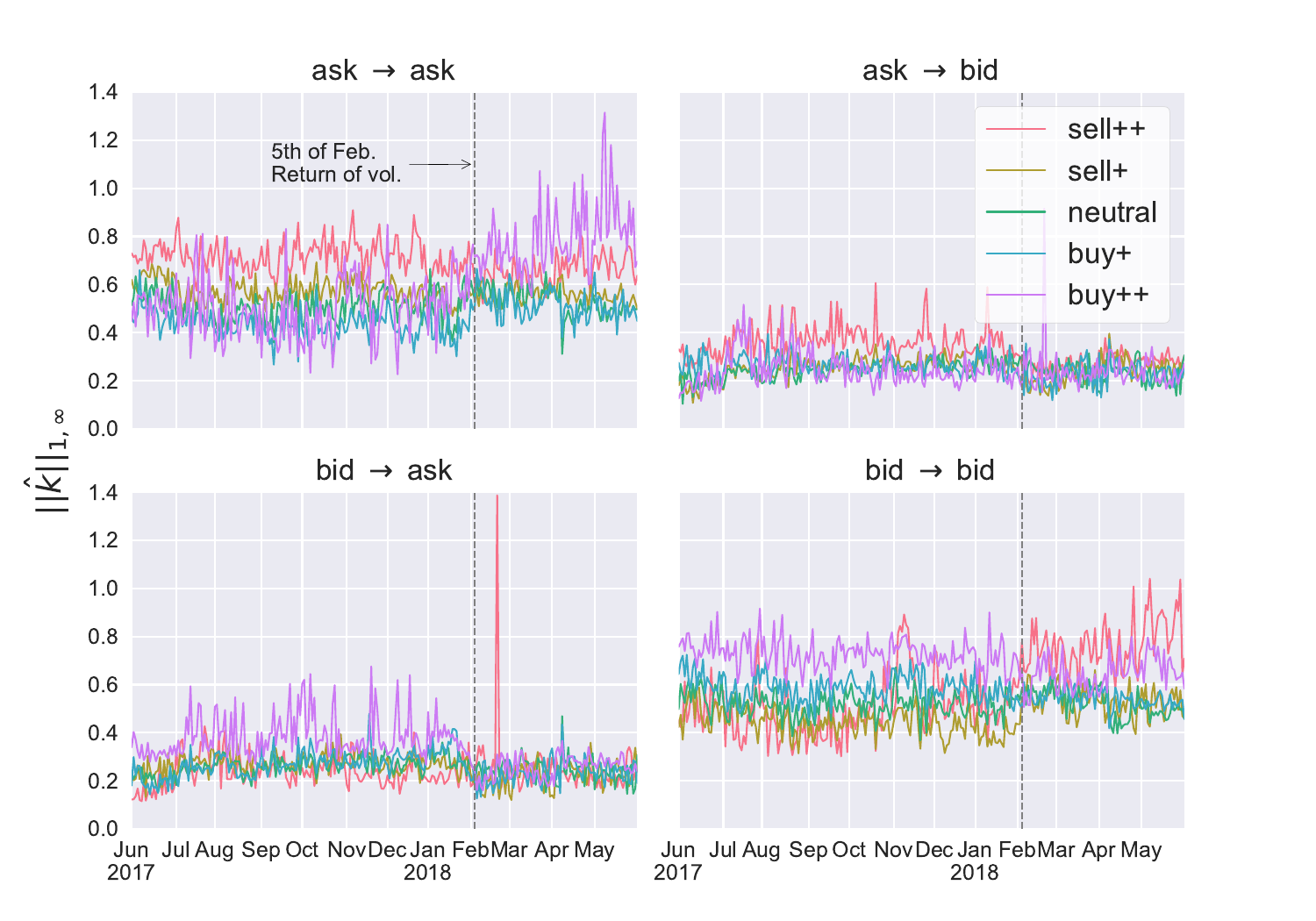}
  \caption{\mikko{Estimated kernel norms $\|\hat \kernels\|_{1,\infty}$ for $\mathrm{Model}_{\mathrm{QI}}$ over time. The dashed line marks 5 February 2018 (``the return of volatility''), a day when the CBOE Volatility Index (VIX) jumped by 116\% to 38 points, a level not seen since August 2015. This day seems to have introduced a systematic change in the magnitude of cross- and self-excitation. The spike in cross-excitation that occurs on 21 February 2018 is linked to a sudden change in market behaviour around 14:00 on that day, when Intel Corporation rolled out patches for its most recent generation of processors. } }
  \label{fig:norms_in_time}
\end{figure}

Even though the \mikko{day-to-day variation} of the estimates is more \mikko{pronounced in} this model \mikko{compared to $\mathrm{Model}_{\mathrm{S}}$,} some clear patterns emerge again.
\mikko{The self-excitation of \event{ask} events consistently increases under heavy sell pressure (state \state{sell++}).} This \mikko{can plausibly be explained by a combination of} a flight to liquidity, \mikko{through the} submission of sell orders, and fear of adverse selection, \mikko{leading to} cancellation\mikko{s} of buy orders, by \mikko{traders} expecting a downwards price move. Under \mikko{mild} buy pressure (state \state{buy+}), \mikko{s}elf-excitation decreases \mikko{so that the corresponding} kernel norm nearly halves.

\mikko{In} the state \state{buy++} (heavy buy pressure), a closer look at the daily estimates of the self-excitation of \event{ask} events plotted in purple in Figure \ref{fig:kernels_imbalance} \mikko{reveals two distinct groups of curves above and below of the median curve.}  \mikko{The daily estimates of the (full) kernel norms plotted in time in Figure \ref{fig:norms_in_time} suggest that the self-excitation effect has undergone a structural break in early February 2018 --- the moment that suddenly marked the end of a year-long period of unusually low volatility in the US equity markets, dubbed ``the return of volatility'' by some financial journalists. It is also worth mentioning that at the same time the price of INTC went above \$50, at which point the} large-tick character of the stock starts to \mikko{weaken.} \mikko{The behaviour characterised by the lower group of curves (pre-February 2018) can be interpreted as traders expecting a price increase and thus deferring the submission of sell orders, whereas the upper group of curves (post-February 2018) hints at a tendency of sellers to seek} an advantageous position in the \mikko{ask} queue. Indeed, a queue imbalance close to \mikko{one implies that} that only a very \mikko{small amount of liquidity is available} at the ask price. Thus, placing a sell limit order following \mikko{a succession of} sell limit orders \mikko{from other traders} allows one to \mikko{acquire} a good position in the queue, should it \mikko{be replenished (which brings} the queue imbalance \mikko{back to equilibrium)}. If, however, the ask queue \mikko{progresses towards depletion}, one has \mikko{still} time to cancel the order whilst the sell limit orders at \mikko{the} front of the queue are matched with incoming buy market orders.

\begin{figure}[t]
  \centering
  \includegraphics[width=.78\linewidth]{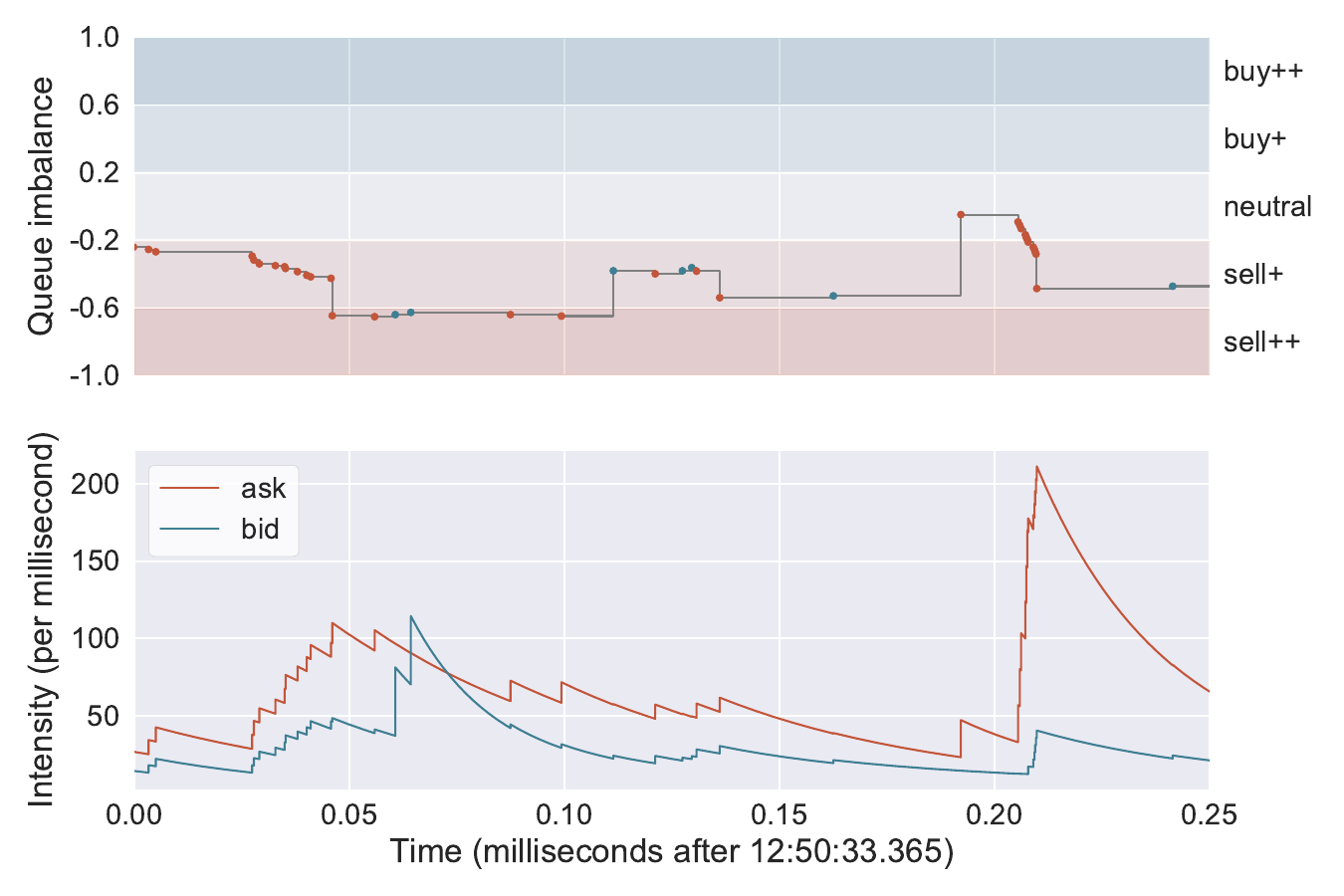}
  \caption{\mikko{The upper panel depicts the evolution of the queue imbalance and level-I order flow of INTC on 13 February 2018 and the \event{ask} (red dots) and \event{bid} (blue dots) events. The lower panel displays the estimated intensities of $\mathrm{Model}_{\mathrm{QI}}$. (The self- and cross-excitation kernel norms of $\mathrm{Model}_{\mathrm{QI}}$ on 13 February 2018 are visualised in Figure \ref{fig:kernels_uncertainty}.)}}
  \label{fig:sample_path}
\end{figure}

\mikko{Under buy pressure (states \state{buy++} and \state{buy+}), the timescale of the cross-excitation from \event{bid} to \event{ask} events becomes almost as short as that of the self-excitation of \event{bid} events. This reflects the resilience of the LOB --- in response to a \event{bid} event, \event{ask} events compete neck and neck with \event{bid} events} precisely when they push \mikko{the queue imbalance back towards the equilibrium state.} \mikko{Recall that the resilience of the LOB} is \mikko{also reinforced} by \mikko{the estimated state transition probabilities, as discussed above.}

\mikko{To exemplify the estimated intensity processes and their state dependence, a very brief extract from the estimated dynamics of $\mathrm{Model}_{\mathrm{QI}}$ is presented in Figure \ref{fig:sample_path}. In particular, we observe} how \mikko{pronounced} the self-excitation of \event{bid} events becomes when the queue imbalance drops below \mikko{$-0.6$, that is, to state \state{sell++}.}

\subsection{Goodness-of-fit diagnostics} \label{subsec:goodness_of_fit}

\mikko{To assess the goodness of fit of the estimated models, we examine the event residuals $\residualOfEvent{e}{n}$ in sample (daily, 12:00--14:30) and out of sample (daily, 14:30--15:00) for both $\mathrm{Model}_{\mathrm{S}}$ and $\mathrm{Model}_{\mathrm{QI}}$.} 
For comparison, we also estimated a\mikko{n ordinary} Hawkes process ($\numberOfStates=1$) with exponential kerne\mikko{l} \mikko{for} the same event types $\eventSpace=\{\mbox{\event{ask}, \event{bid}}\}$ and computed \mikko{its} event residuals. \mikko{Since the state-dependent Hawkes process nests the ordinary Hawkes process, the former will by construction} provide a better fit \mikko{in sample than the latter.} The \mikko{Q--Q} plots in Figure \ref{fig:goodness_of_fit} show that this improvement in \mikko{the goodness of fit extends} out of sample, albeit \mikko{the improvement is smaller} than in sample. \mikko{This is a confirmation that} that \mikko{$\mathrm{Model}_{\mathrm{S}}$ and $\mathrm{Model}_{\mathrm{QI}}$, and their state-dependent features in particular,} are not overfitted.

It should not come as a surprise that the improvement in \mikko{goodness of fit} provided by the state-dependent model looks \mikko{meagre} when one examines the \mikko{Q--Q} plots. \mikko{Indeed, the behaviou\mikko{r} of $\mathrm{Model}_{\mathrm{S}}$, $\mathrm{Model}_{\mathrm{QI}}$ and their ordinary Hawkes process alternative is quite similar when the bid--ask spread and queue imbalance are in their most likely states.}
It is only \mikko{in} the less likely states, \state{2+} in $\mathrm{Model}_{\mathrm{S}}$ and \state{sell++} and \state{buy++} in $\mathrm{Model}_{\mathrm{QI}}$, \mikko{where the difference between the state-dependent and ordinary Hawkes process models becomes more pronounced.}
Thereby, the \mikko{unconditional} distribution of \mikko{residuals does not vary much between these three models.}

\begin{figure}[t]
  \begin{subfigure}[t]{.5\textwidth}
  \centering
  \includegraphics[width=.95\linewidth]{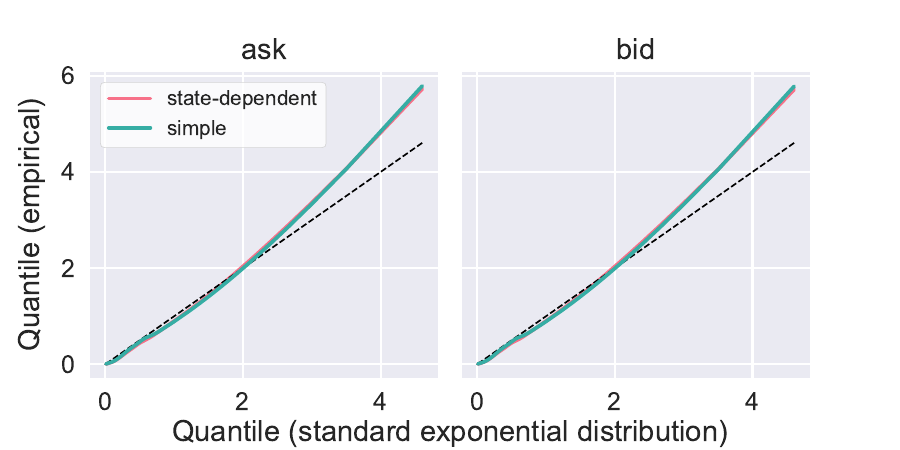}
  \caption{\maxime{$\mathrm{Model}_{\mathrm{S}}$: in sample.}}
  \label{fig:qq_model_S_in_sample}
\end{subfigure}
\begin{subfigure}[t]{.5\textwidth}
  \centering
  \includegraphics[width=.95\linewidth]{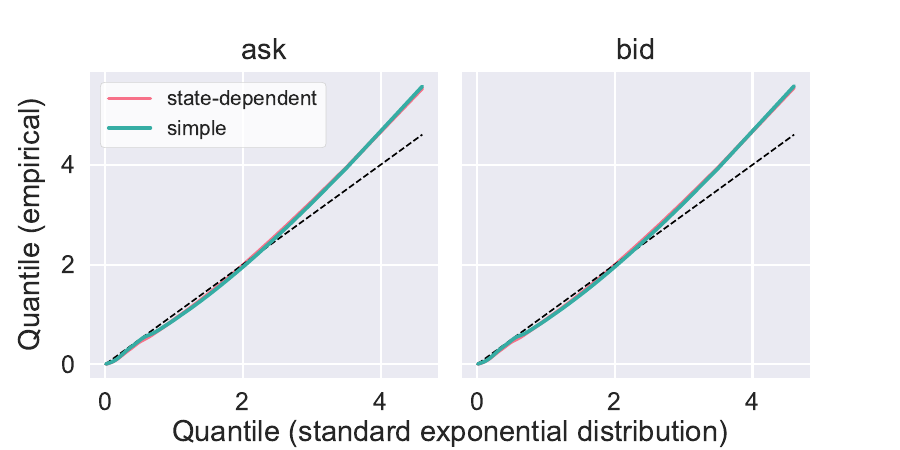}
  \caption{\maxime{$\mathrm{Model}_{\mathrm{S}}$: out of sample.}}
  \label{fig:qq_model_S_out_sample}
\end{subfigure}
\\
\begin{subfigure}[t]{.5\textwidth}
  \centering
  \includegraphics[width=.95\linewidth]{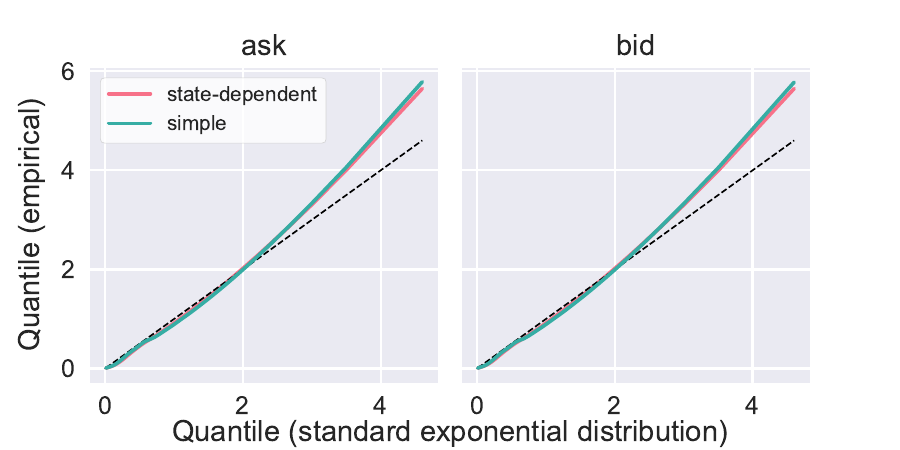}
  \caption{\maxime{$\mathrm{Model}_{\mathrm{QI}}$: in sample.}}
  \label{fig:qq_model_QI_in_sample}
\end{subfigure}
\begin{subfigure}[t]{.5\textwidth}
  \centering
  \includegraphics[width=.95\linewidth]{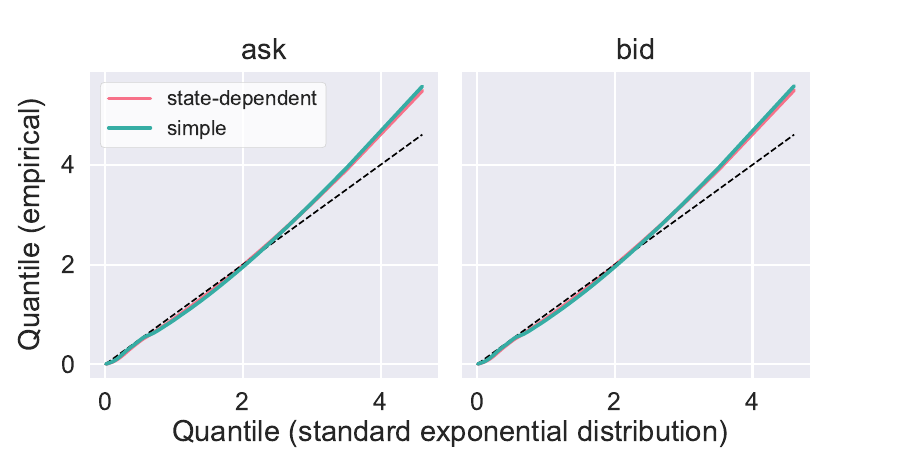}
  \caption{\maxime{$\mathrm{Model}_{\mathrm{QI}}$: out of sample.}}
  \label{fig:qq_model_QI_out_sample}
\end{subfigure}
  \caption{\maxime{In-sample (12:00--14:30) and out-of-sample (14:30--15:00) Q--Q plots of event residuals under $\mathrm{Model}_{\mathrm{S}}$, $\mathrm{Model}_{\mathrm{QI}}$ (state-dependent) and an ordinary Hawkes process (simple). The residuals of the $i$th day are computed using the ML estimates $(\hat\baseRates^{(i)}, \hat\impactCoeffs^{(i)}, \hat\decayCoeffs^{(i)})$ obtained from the 12:00--14:30 period. The empirical quantiles are obtained by pooling the residuals of all 250 trading days. The two panels in each sub-figure correspond to the sequences of residuals $(\residualOfEvent{e}{n})$ for $e\in\eventSpace = \{\text{\event{ask}},\text{\event{bid}} \}$.}}
  \label{fig:goodness_of_fit}
\end{figure}

\subsection{\mikko{Endogeneity is state-dependent}}\label{sec:endo}

\mikko{The recent popularity of Hawkes processes in the modelling of high-frequency financial data partly stems from their ability to quantify the \emph{endogeneity} of market activity. Indeed, based on the cluster representation of Hawkes processes \citep{Hawkes:1974aaClusterRepresentation} and the theory of branching processes \citep{HarrisTheodoreE1963Ttob}, the expected number of new events triggered by each event through self-excitation in a univariate ordinary Hawkes process equals the $L^1$-norm of its self-excitation kernel, provided it is less than one. The threshold one is the critical boundary for the stability of the process and validity of the cluster representation.} \mikko{Hawkes processes fitted to high-frequency financial data often exhibit kernel norms slightly below one, whilst the base rates tend to be relatively low in comparison. This phenomenon has been interpreted by \citet{Filimonov2012} and \citet{Hardiman2013} as evidence that most market events are endogenous, mere responses to earlier events, dwarfing the flow of less frequent exogenous events that are driven by new information. They dub the phenomenon \emph{critical reflexivity}, which is a nod to George Soros and his reflexivity theory on the endogeneity of financial markets \citep{Soros1989}.}

\mikko{In the context of state-dependent Hawkes processes,} the kernel $\collection{\kernelAtoB{e'}{e}(\cdot, x)}{e',e\in\eventSpace}$, for any \mikko{state} $x\in\stateSpace$, defines \mikko{a multivariate ordinary} Hawkes process. \mikko{Whilst in the multivariate case there is no direct analogue of the kernel norm --- at least none with an equally clear-cut interpretation --- }
 the spectral radius $\rho(x)$ of the $\numberOfEvents\times\numberOfEvents$ matrix $(m_{ij})$, where $m_{ij}:=\|\kernelAtoB{j}{i}(\cdot,x)\|_{1,\infty}$, \mikko{can be understood as a measure of endogeneity in this ordinary Hawkes process for each $x\in\stateSpace$.} \mikko{It also characterises} the \mikko{stability} of \mikko{the} process\mikko{, whereby} $\rho(x)<1$ \mikko{is a sufficient condition for} the existence of and convergence to a stationary version \citep{Bremaud:1996aa:StabilityNonLinearHawkes}. Figure \ref{fig:radii} displays the daily estimates of $\rho(x)$ for both \mikko{$\mathrm{Model}_{\mathrm{S}}$ and $\mathrm{Model}_{\mathrm{QI}}$ as a function of $x \in \stateSpace$.} \mikko{We observe a remarkably clear pattern of $\rho(x)$ being almost uniformly higher in the disequilibrium states (\state{2+}, \state{sell++}, \state{buy++}) than in the equilibrium states (\state{1}, \state{sell+}, \state{neutral}, \state{buy+}). In particular, in $\mathrm{Model}_{\mathrm{S}}$,} the spectral radius $\rho(\mbox{\state{2+}})$ is systematically above the critical value 1\mikko{, whilst in $\mathrm{Model}_{\mathrm{QI}}$, the values of $\rho(\mbox{\state{sell++}})$ and $\rho(\mbox{\state{buy++}})$ are above 0.9 half of the time, exceeding 1 occasionally.} \mikko{These results thus open a new perspective on critical reflexivity, showing that it is in fact a largely state-dependent phenomenon, observed only in particular circumstances. They also lend credence to Soros's remark that \emph{``[e]ven in the financial markets demonstrably reflexive processes occur only intermittently''} \citep[p.~29]{soros2009crash}.
}

\begin{figure}[t]
\begin{subfigure}[t]{0.45\textwidth}
  \centering
  \includegraphics[width=.67\linewidth]{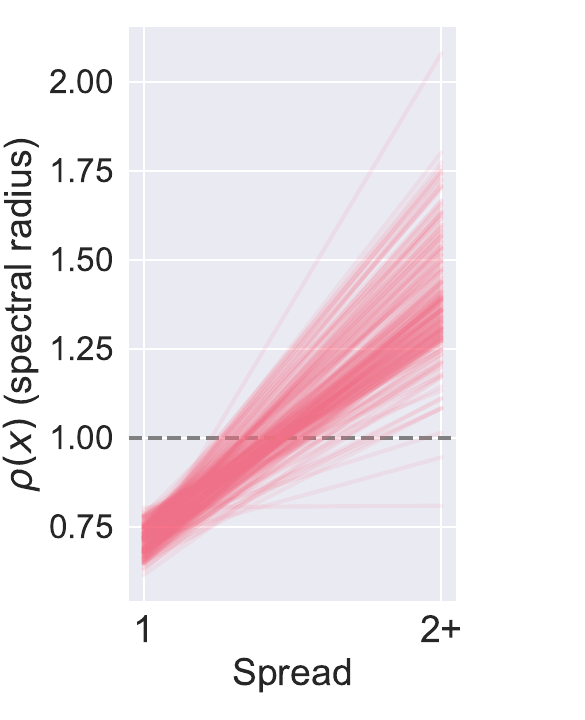}
  \caption{\mikko{$\mathrm{Model}_{\mathrm{S}}$}}
  \label{fig:radii_S}
\end{subfigure}
\begin{subfigure}[t]{0.55\textwidth}
  \centering
  \includegraphics[width=.9\linewidth]{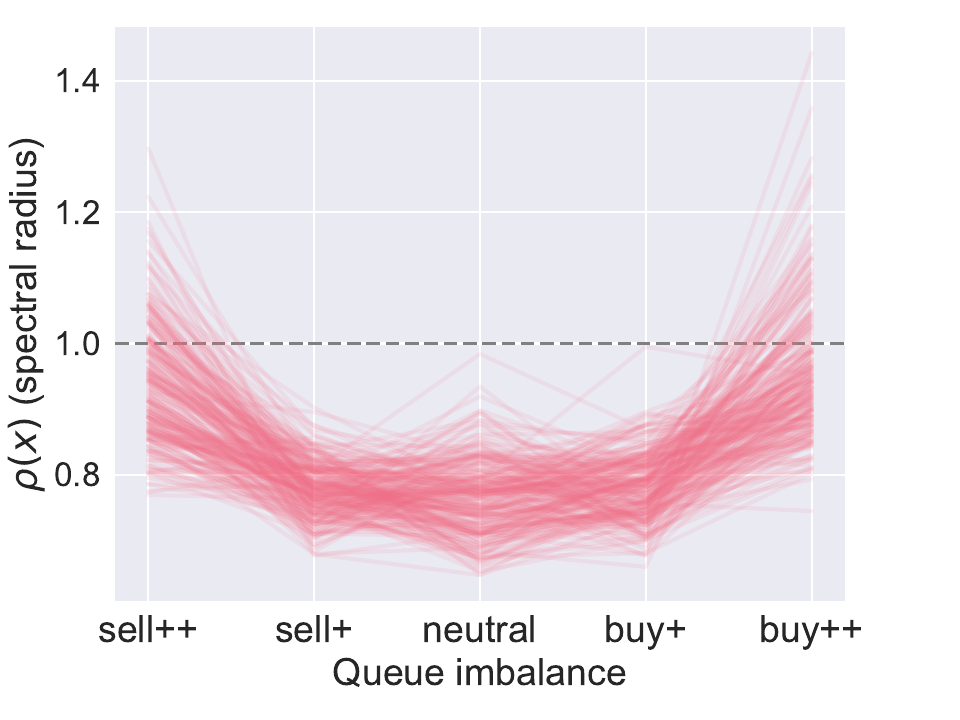}
  \caption{\mikko{$\mathrm{Model}_{\mathrm{QI}}$}}
  \label{fig:radii_QI}
\end{subfigure}
\caption{\mikko{The estimated spectral radius $\hat\rho(x)$ as a function of $x\in\stateSpace$ under $\mathrm{Model}_{\mathrm{S}}$ and $\mathrm{Model}_{\mathrm{QI}}$.
The daily profiles $x\mapsto \hat\rho(x)^{(i)}$, $i=1,\ldots,250$, are represented by the red translucent curves.}}
\label{fig:radii}
\end{figure}

\mikko{The increase in endogeneity in the disequilibrium} states \mikko{seems attributable to strategies employed by high-frequency traders (HFTs), who become active in these states, in anticipation of a price move. In a recent study, \citet{Lehalle:2018:SignalsOptimalTrading} analyse a unique data set from Nasdaq Stockholm, where the identities of the buyer and seller in each transaction were disclosed by the exchange until 2014. In particular, they show that when the queue imbalance increases, the trading activity of market participants they classify as proprietary HFTs is amplified, in the direction of the imbalance. The pronounced sub-millisecond self-excitation effects seen in Figure \ref{fig:kernels_imbalance}, which are the key driver behind the high spectral radii $\rho(x)$, concur with the trading patterns observed by \citet{Lehalle:2018:SignalsOptimalTrading}. Besides its use as a trading signal, \citet{Lehalle:2018:SignalsOptimalTrading} find the queue imbalance to be mean-reverting, which is similarly compatible with our results. A higher spectral radius in disequilibrium states corresponds here to an increase in market activity, which, reinforced by the structure of the estimated transition probabilities (Figure \ref{fig:transition_probabilities}), is more likely to push the queue imbalance towards equilibrium than vice versa.  
}

\subsection{\mikko{Event--state} structure of limit order books}

\mikko{We could alternatively build a state-dependent variant of a Hawkes process in the following, conceptually simpler way.} Using the representation $\Nhybrid$ of the counting and state processes $(\N, X)$, one could specify an intensity of the form
\begin{equation}  \label{eq:intensity_alternative}
	\intensityOfEventState{e}{x}(t) = \nu_{ex} + \sum_{e'\in\eventSpace, x'\in\stateSpace} \int_{[0,t)}\kernelAtoB{e'x'}{ex}(t-s)d\NhybridOfEventState{e'}{x'}(s),\quad t\geq 0, \quad e\in\eventSpace,\quad x\in\stateSpace,
\end{equation}
instead of \eqref{eq:hybrid_intensity}. \mikko{This approach would in fact be tantamount to simply using} a $\numberOfEvents\numberOfStates$-dimensional \mikko{ordinary} Hawkes process.

\mikko{The intensity \eqref{eq:intensity_alternative} makes self- and cross-excitation state-dependent,} but the \mikko{simple} structure of the state process \mikko{$X$} in Definition \ref{def:state_dependent_hawkes} is lost. \mikko{Namely, under} \eqref{eq:intensity_alternative}, the transition probabilities of the state process depend \mikko{not only} on the current state but on the entire history. However, LOBs \mikko{enjoy} a certain \mikko{\emph{event--state structure ---}} knowing the current state of the LOB and the characteristics of the next order suffices to (approximately) determine the next state. State-dependent Hawkes processes are by \mikko{construction} able to \mikko{reproduce} such an event-state structure and, therefore, compared to the alternative \eqref{eq:intensity_alternative}, we expect them to provide in general a better statistical \mikko{description of the LOB}. Moreover, \mikko{the} model given by \eqref{eq:intensity_alternative} requires \mikko{a kernel with} $\numberOfEvents^2\numberOfStates^2$ \mikko{components} whereas a state-dependent Hawkes process can be specified \mikko{more parsimoniously}, using only $\numberOfEvents^2\numberOfStates$ \mikko{components}.

\begin{figure}[t]
  \centering
  \includegraphics[width=1\linewidth]{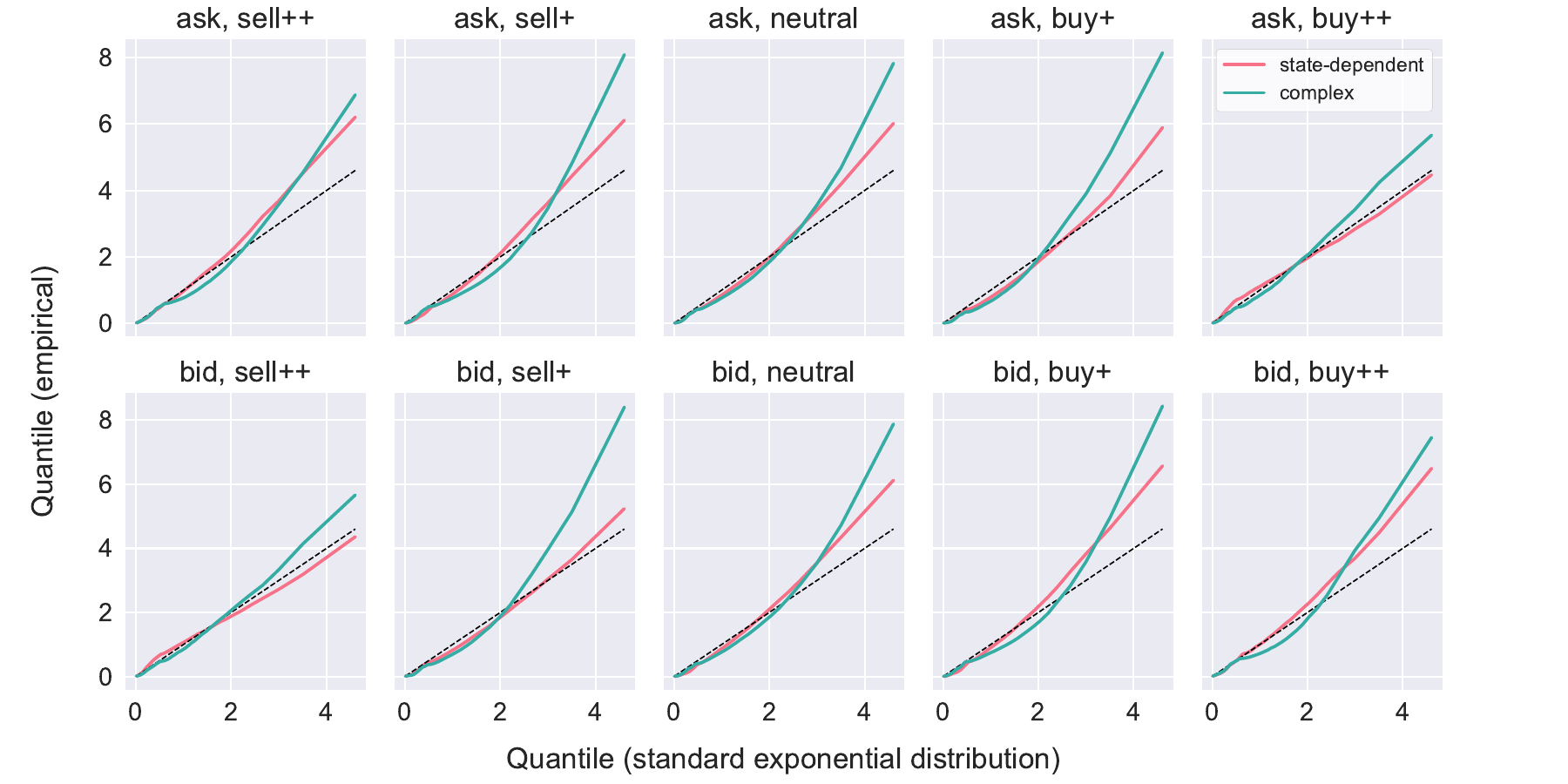}
  \caption{In-sample (\mikko{12:00--14:30}) \mikko{Q--Q} plots of total residuals \mikko{of $\mathrm{Model}_{\mathrm{QI}}$} (state-dependent) and the alternative model given by \eqref{eq:intensity_alternative} (complex) \mikko{on 11} May 2018. \mikko{Each panel corresponds to a} sequence of residuals $(\residualOfEventState{e}{x}{n})$ \mikko{for all $e\in\eventSpace$ and $x\in\stateSpace$.}}
  \label{fig:qq_complex}
\end{figure}

\mikko{To compare this alternative model to state-dependent Hawkes processes, we estimate an ordinary Hawkes process with intensity \eqref{eq:intensity_alternative}, choosing an exponential form of the kernel $\kernels$ and specifying the event types in $\eventSpace$ and the state space $\stateSpace$ as in $\mathrm{Model}_{\mathrm{QI}}$.}
Because of the higher dimensionality of this alternative model, estimation becomes computationally more expensive,  
\mikko{which is why we use INTC data for May 2018 only.}
In Figure \ref{fig:qq_complex}, this alternative model is compared to \mikko{$\mathrm{Model}_{\mathrm{QI}}$} via \mikko{Q--Q} plots of the total residuals. We \mikko{choose} to display the results for a day when the alternative model provides one of its best fits, \mikko{although the goodness of fit does not} vary \mikko{markedly} across the 22 trading days \mikko{in May 2018}. Even though it has \mikko{fewer} parameters (92 \mikko{as opposed to} 210), we observe that \mikko{$\mathrm{Model}_{\mathrm{QI}}$} provides a better in-sample fit. These empirical results \mikko{thus underline the significance of}  the \mikko{event--state} structure of LOBs.

\section{Discussion}

State-dependent Hawkes processes \mikko{enable us to} \mikko{model} two-way interaction between a self- and cross-exciting point process, \mikko{governing the temporal flow of \emph{events}}, and a state process, \mikko{describing} a \mikko{\emph{system}}. \mikko{In the context of} LOB modelling, they \mikko{provide a probabilistic foundation for a novel class of continuous-time models that encapsulate the} feedback loop between the order flow and \mikko{the shape of the LOB.} \mikko{Our estimation results, using these models and one year's worth of high-frequency LOB data on the stock of Intel Corporation, reveal that state dependence is indeed significant, as we uncover several robust patterns that persist throughout the daily estimation results. In particular, we find that market endogeneity, measured through the magnitude of self- and cross-excitation is state-dependent, being most pronounced in disequilibrium states of the LOB.}

\mikko{Our results also validate the event--state structure of LOBs} that is embedded in the definition of \mikko{a state-dependent Hawkes process}. However, we do not claim that \mikko{$\mathrm{Model}_{\mathrm{S}}$ and $\mathrm{Model}_{\mathrm{QI}}$ would be} the best possible representations of the \mikko{aforementioned} feedback loop \mikko{--- we recognise that they could be refined as follows:}
\begin{enumerate}[label=\textup{(\roman*)}]
\item\label{item:ref1} \mikko{The exponential excitation kernels could be replaced by power laws, motivated by the non-parametric estimation results on ordinary Hawkes processes \citep{Hardiman2013, bacry2016estimation}.} A \mikko{numerically more convenient} alternative \mikko{would be to} use a linear combination of \mikko{exponentials} within the parametric framework to mimic \mikko{slow power-law decay} \citep{Rambaldi2015, Lu2018}.
\item \mikko{The base rates could be made state-dependent by replacing $\baseRate{e}$ with $\baseRate{e}(X(t))$ in \eqref{eq:event_intensity}.} Note that in this case, the model would contain both a continuous-time Markov chain ($\kernels\equiv 0$) and \mikko{an ordinary} Hawkes process as \mikko{special} cases.
\item \mikko{One might argue for excitation kernels that allow for negative values, to capture inhibition effects that are known to exist in LOB data \citep{Lu2018}. To ensure the non-negativity of the intensity, this would require transforming the right-hand side of \eqref{eq:event_intensity} by a non-linear function.}
\item\label{item:ref4} \mikko{More} granular event types and states \mikko{would provide} more nuanced understanding of the LOB dynamics. Besides, notice that the present framework \mikko{accommodates} multiple state variables. For example, $X$ could in fact \mikko{jointly} represent both the \mikko{bid--ask} spread and the queue imbalance \mikko{using} the state space
\begin{equation*}
\stateSpace=\{ \mbox{(\state{1}, \state{sell++})},\ldots, \mbox{(\state{1}, \state{buy++})}, \mbox{(\state{2+}, \state{sell++})}, \ldots, \mbox{(\{\state{2+}, \state{buy++})}\}.
\end{equation*}
\end{enumerate}

\mikko{However}, no matter how one modifies the intensity \eqref{eq:event_intensity} \mikko{by implementing any of \ref{item:ref1}--\ref{item:ref4}} to incorporate one's views on the feedback loop, the \mikko{event--state} structure of the model \mikko{remains intact}. The process still falls within the class of hybrid marked point processes \citep{morariu:2017:hybrid}, \mikko{implying} that the theoretical results of this paper (existence, uniqueness, \mikko{separability of the} likelihood \mikko{function}) still \mikko{apply}.

\appendix
\section{Appendix} \label{sec:appendix}

\subsection{Proofs}\label{sec:proofs}

\begin{proof}[Proof of Theorem \ref{thm:characterisation}]
	The statement \mikko{follows} by applying Theorem 2.13 in \citet{morariu:2017:hybrid}. We only need to check that if $(\N, X)$ is a state-dependent Hawkes process, then $\Nhybrid$ admits an $\internalHistoryCollection{\Nhybrid}$-intensity. By Proposition 3.1 in \citet{jacod:1975:projection}, $\Nhybrid$ admits \mikko{an} $\internalHistoryCollection{\Nhybrid}$-compensator $\Lambda$ given by
	\begin{equation*}
		\Lambda(dt,de,dx) = \sum_{n\in\integers}\frac{G_n(dt-T_n, de, dx)}{G_n([t-T_n,\infty],\eventSpace,\stateSpace)}\indicator_{\{T_n<t\leq T_{n+1}\}},
	\end{equation*}
	 and thus this holds if the conditional distributions $G_n$ of $(T_{n+1}-T_{n}, E_{n+1}, X_{n+1})$ with respect to $\internalHistory{\Nhybrid}{T_n}$ are absolutely continuous with respect to $dt\mu_{e}(de)\mu_{x}(dx)$  on $(0, T_{n+1}-T_n]\times\eventSpace\times\stateSpace$, where $\mu_e$ and $\mu_x$ are the counting measures on $\eventSpace$ and $\stateSpace$, respectively. Moreover, by definition of $\Nhybrid$, $\Lambda(dt, de, \stateSpace)$ is \mikko{an} $\internalHistoryCollection{\Nhybrid}$-compensator \mikko{of} $\N$. But since $\N$ admits an $\internalHistoryCollection{\Nhybrid}$-intensity, by uniqueness of the compensator \citep[Theorem 1.25, p.~39]{kallenberg2017random}, we \mikko{necessarily have} with probability one that
	\begin{equation*}
		\sum_{x\in\stateSpace}G_n(dt-T_n, \{e\}, \{x\})\indicator_{\{T_n<t\leq T_{n+1}\}} = f(t,e)dt\indicator_{\{T_n<t\leq T_{n+1}\}},\quad e\in\eventSpace, \quad n\in\integers,
	\end{equation*}
	for some $\internalHistory{\Nhybrid}{T_n}\otimes\mathcal{B}(\realLine) \otimes\mathcal{P}(\eventSpace)$-measurable function $f$, which \mikko{concludes the proof}.
\end{proof}

\begin{proof}[Proof of Theorem \ref{thm:existence_uniqueness}]
	By Theorem \ref{thm:characterisation}, it is sufficient to show that each of the two conditions ensures the existence and uniqueness of a non-explosive point process $\Nhybrid$ with an $\internalHistoryCollection{\Nhybrid}$-intensity given by \eqref{eq:hybrid_intensity}, which is achieved by applying Theorems 2.17 and 2.21 in \citet{morariu:2017:hybrid}.
\end{proof}

\begin{proof}[Proof of Theorem \ref{thm:likelihood_separability}]
	By Theorem \ref{thm:characterisation}, we know that $\Nhybrid$ \mikko{has} intensity $\intensityOfEventsStates$, which is given by \eqref{eq:hybrid_intensity}. Hence, by applying Proposition 7.3.III in \citet[p.~251]{daleyVereJonesVolume1}, we can express the \mikko{log likelihood function} as
\begin{equation} \label{eq:log-likelihood-formula}
	\ln \likelihood (\transitionProbabilities, \baseRates, \kernelsParams) = \sum_{n=1}^{N}\ln\intensityOfEventState{e_n}{x_n}(t_n) - \int_0^T\sum_{e\in\eventSpace, x\in\stateSpace}\intensityOfEventState{e}{x}(t) dt.
\end{equation}
Plugging \eqref{eq:hybrid_intensity} in \eqref{eq:log-likelihood-formula} and using that $\sum_{x'\in\stateSpace}\transitionProbabilityIfEventFromAtoB{e}{x}{x'}=1$, $e\in\eventSpace$, $x\in\stateSpace$, yields \eqref{eq:log-likelihood}, from which it is immediate that $(\hat\transitionProbabilities,\hat\baseRates,\hat\kernelsParams)\in\argmax_{\transitionProbabilities,\baseRates,\kernelsParams}\likelihood(\transitionProbabilities,\baseRates,\kernelsParams)$ if and only if
	\begin{equation*}
		\begin{dcases}
		 \hat\transitionProbabilities &\in \argmax_{\transitionProbabilities} \sum_{n=1}^{N}\ln\transitionProbabilityIfEventFromAtoB{e_n}{x_{n-1}}{x_n}  , \\
		(\hat \baseRates, \hat \kernelsParams ) &\in \argmax_{\baseRates,\kernelsParams } \sum_{n=1}^{N}\ln \intensityOfEvent{e_n}(t_n) - \int_0^T\sum_{e\in\eventSpace}\intensityOfEvent{e}(t)dt,
	\end{dcases}
	\end{equation*}
	where the first optimisation problem is performed under the constraint that $\transitionProbabilityIfEvent{e}$ is a transition probability matrix, $e\in\eventSpace$. By solving this optimisation problem with the method of Lagrange multipliers, we obtain the claimed expression for $\hatTransitionProbabilityIfEventFromAtoB{e}{x}{x'}$.
\end{proof}

\begin{proof}[Proof of Theorem \ref{thm:residuals}]
	The statement follows directly from Theorem 1 in \citet{Brown:1988aa}. Note that this theorem requires that, with probability one, $\int_0^t\intensityOfEvent{e}(s)ds\rightarrow\infty$ and $\int_0^t\intensityOfEventState{e}{x}(s)ds\rightarrow\infty$ as $t\rightarrow\infty$, $e\in\eventSpace$, $x\in\stateSpace$. The first condition is satisfied because, in Definition \ref{def:state_dependent_hawkes}, we assume that all the base rates are strictly positive ($\baseRates\in\realLinePositive^{\numberOfEvents}$). By Lemma 17 in \citet[p.~41]{bremaud:1981:MartingaleDynamics}, the second condition is equivalent to $\NhybridOfEventState{e}{x}(t)\rightarrow\infty$, $t\rightarrow\infty$, with probability one, $e\in\eventSpace$, $x\in\stateSpace$, which is assumed here.
\end{proof}

\subsection{Details on maximum likelihood estimation}

\subsubsection{Formulae for state-dependent Hawkes processes with exponential kernels} \label{subsec:formulae_liklihood}

In the case of exponential \mikko{kernel} $\kernels$ given by \eqref{eq:kernels_exp}, the following formulae can be derived for the second and third terms of the \mikko{log likelihood function} in \eqref{eq:log-likelihood}, denoted by $l_+$ and $l_-$, respectively. Here, we consider a general time horizon $(t_0, T]$, meaning that the origin of time is not necessarily $t_0=0$ and that the times $t^{ex}_n \leq t_0$ are treated like an initial condition\mikko{, and we have}
\begin{align*}
l_+ &= \sum_{e}\sum_{t^e_n >t_0}\ln \intensityOfEvent{e}(t_n^e) = \sum_{e}\sum_{t^e_n >t_0}\ln\left( \nu_e + \sum_{e'x'}\alpha_{e'x'e}S_{e'x'e}(t^e_n) \right), \nonumber \\
S_{e'x'e}(s,t) :&= \sum_{s<t^{e'x'}_i<t}\exp(-\beta_{e'x'e}(t-t^{e'x'}_i)), \quad S_{e'x'e}(t) := S_{e'x'e}(-\infty,t), \nonumber \\
S_{e'x'e}(t) &= \exp(-\beta_{e'x'e}(t-s))S_{e'x'e}(-\infty, s] + S_{e'x'e}(s, t), \label{eq:partial_sum_recursion} \\
l_- &= \int_{t_0}^T\sum_e \intensityOfEvent{e}(t)dt = \int_{t_0}^T\sum_e \left(  \nu_e + \sum_{e'x'}\sum_{t^{e'x'}_i < t}\alpha_{e'x'e}\exp(-\beta_{e'x'e}(t-t^{e'x'}_i) \right)dt \nonumber \\
&= \sum_{e} \nu_{e} (T-t_0) + \sum_{e'x'}\sum_{t^{e'x'}_i\leq t_0}\frac{\alpha_{e'x'e}}{\beta_{e'x'e}}(\exp(-\beta_{e'x'e}(t_0-t^{e'x'}_i))-\exp(-\beta_{e'x'e}(T-t^{e'x'}_i)) \nonumber \\
&\qquad + \sum_{t^{e'x'}_i> t_0} \frac{\alpha_{e'x'e}}{\beta_{e'x'e}}(1-\exp(-\beta_{e'x'e}(T-t^{e'x'}_i)). \nonumber
\end{align*}
The gradients can then also be computed \mikko{via}
\begin{align*}
\frac{\partial l_+}{\partial \nu_k} &= \sum_{t^k_n >t_0}\left( \nu_k + \sum_{e'x'}\alpha_{e'x'k}S_{e'x'k}(t^k_n)  \right)^{-1}, \\
\frac{\partial l_+}{\partial \alpha_{ijk}} &= \sum_{t^k_n >t_0}\frac{S_{ijk}(t_n^e)}{\nu_k + \sum_{e'x'}\alpha_{e'x'k}S_{e'x'k}(t^k_n)}, \\
\frac{\partial l_+}{\partial \beta_{ijk}} &= - \sum_{t^k_n >t_0}\frac{\alpha_{ijk}S_{ijk}^{(1)}(t^e_n)}{\nu_k + \sum_{e'x'}\alpha_{e'x'k}S_{e'x'k}(t^k_n)}, \\
S_{e'x'e}^{(1)}(s,t) :&= \sum_{s<t^{e'x'}_i<t}(t - t^{e'x'}_i)\exp(-\beta_{e'x'e}(t-t^{e'x'}_i)), \quad S_{e'x'e}^{(1)}(t) = S_{e'x'e}^{(1)}(-\infty, t), \\
S_{e'x'e}^{(1)}(t) &= \exp(-\beta_{e'x'e}(t-s))\left(S^{(1)}_{e'x'e}(-\infty, s] +  (t-s)S_{e'x'e}(-\infty, s]\right) + S_{e'x'e}^{(1)}(s, t), \\
\frac{\partial l_-}{\partial \nu_k} &= T - t_0, \\
\frac{\partial l_-}{\partial \alpha_{e'x'e}} &= \sum_{t^{e'x'}_i\leq t_0}\frac{1}{\beta_{e'x'e}}(\exp(-\beta_{e'x'e}(t_0-t^{e'x'}_i))-\exp(-\beta_{e'x'e}(T-t^{e'x'}_i)) \\
&\qquad+ \sum_{t^{e'x'}_i> t_0} \frac{1}{\beta_{e'x'e}}(1-\exp(-\beta_{e'x'e}(T-t^{e'x'}_i)), \\
\frac{\partial l_-}{\partial \beta_{e'x'e}} &= \sum_{t^{e'x'}_i\leq t_0}\frac{\alpha_{e'x'e}}{\beta_{e'x'e}}((T-t^{e'x'}_i)\exp(-\beta_{e'x'e}(T-t^{e'x'}_i))-(t_0-t^{e'x'}_i)\exp(-\beta_{e'x'e}(t_0-t^{e'x'}_i))) \\
&\qquad-\frac{\alpha_{e'x'e}}{\beta_{e'x'e}^{2}}(\exp(-\beta_{e'x'e}(t_0-t^{e'x'}_i))-\exp(-\beta_{e'x'e}(T-t^{e'x'}_i)) \\
&\qquad+ \sum_{t^{e'x'}_i> t_0} \frac{\alpha_{e'x'e}}{\beta_{e'x'e}}(T-t^{e'x'}_i)\exp(-\beta_{e'x'e}(T-t^{e'x'}_i)) - \frac{\alpha_{e'x'e}}{\beta_{e'x'e}^2}(1-\exp(-\beta_{e'x'e}(T-t^{e'x'}_i))).
\end{align*}
Besides, the efficient computation of the residuals is based on the \mikko{identity}
\begin{align*}
	\int_s^t \intensityOfEvent{e}(u)du = & \,\, (t-s)\baseRate{e} + \sum_{e'x'}C_{e'x'e}(t) - C_{e'x'e}(s) - \sum_{e'x'}\frac{\impactCoeff{e'}{x'}{e}}{\decayCoeff{e'}{x'}{e}}(S_{e'x'e}(t)-S_{e'x'e}(s)), \\
	C_{e'x'e}(t) := &\,\, \sum_{t^{e'x'}_i<t}\frac{\impactCoeff{e'}{x'}{e}}{\decayCoeff{e'}{x'}{e}}\exp(-\decayCoeff{e'}{x'}{e}(t_0\vee t^{e'x'}_i - t^{e'x'}_i)).
\end{align*}

\mikko{The accompanying Python library}  \texttt{mpoints} implements the above formulae in C via the Cython \mikko{extension of} Python, allowing us to drastically reduce the computation time (up to 300 times faster computations compared to a \mikko{plain} Python implementation \mikko{using NumPy}). This played a crucial role in \mikko{making} the present study \mikko{computationally feasible.}

\begin{figure}[!t]
\begin{subfigure}[t]{.5\textwidth}
  \centering
  \includegraphics[width=.99\linewidth]{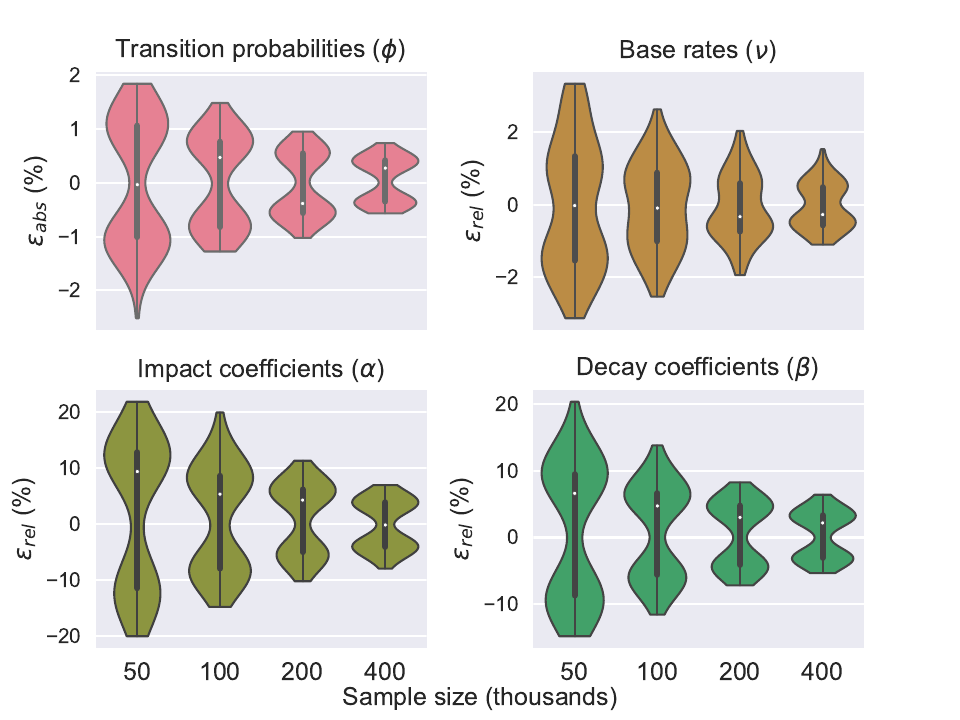}
  \caption{Specification 1.}
  \label{fig:convergence0}
\end{subfigure}
\begin{subfigure}[t]{.5\textwidth}
  \centering
  \includegraphics[width=.99\linewidth]{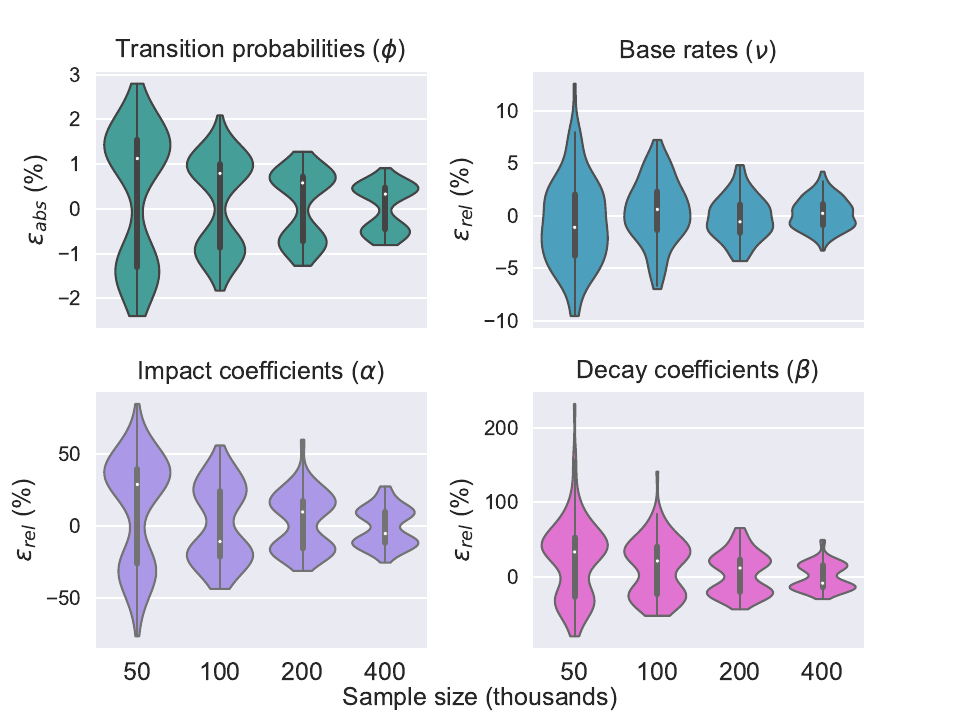}
  \caption{Specification 2.}
  \label{fig:convergence1}
\end{subfigure}
\caption{Violin plots of the worst estimation errors \eqref{eq:worst_error} under two different sets of parameter values (specifications 1 and 2). For every specification and sample size $N$, we simulate 100 paths with sample size $N$ and perform \mikko{ML} estimation for each of them. The true parameters are used as the initial guess in the optimisation procedure to speed up \mikko{estimation} and reduce the computation\mikko{al cost}.}
\label{fig:convergence}
\end{figure}

\subsubsection{Numerical optimisation}\label{app:numerical_opt}

Similarly to \mikko{ordinary} Hawkes processes \citep{Lu2018}, the maximum likelihood (ML) search \eqref{eq:optimisation_problem} is broken down into $d_e$ separate optimisation problems. For every $e\in\eventSpace$, we have an independent optimisation problem that involves only the parameters $\baseRate{e}$, $\collection{\impactCoeff{e'}{x}{e}}{e'\in\eventSpace, x\in\stateSpace}$, $\collection{\decayCoeff{e'}{x}{e}}{e'\in\eventSpace, x\in\stateSpace}$, and which we solve \mikko{using} a \mikko{conjugate gradient} method.

More precisely, we call the \mikko{\texttt{minimize}} function in the \mikko{\texttt{scipy.optimize}} Python package and use the method \mikko{\texttt{TNC}}. Three random sets of parameters are generated and used as \mikko{alternative} initial guesses. \mikko{An ordinary} Hawkes process ($\numberOfStates=1$) is estimated before the state-dependent one and its parameters are also used as an initial guess. Moreover, the estimate of the previous day is used as another initial guess and, thus, a total of five different initial guesses are employed. For \mikko{$\mathrm{Model}_{\mathrm{QI}}$}, around 50 iterations and 400 function evaluations are required for \mikko{each day, event type $e\in\eventSpace$ and initial guess.}

As this optimisation problem is non-convex, the above procedure might \mikko{converge to a mere} local maximum instead of the global one \citep{Lu2018}. Nevertheless, \mikko{in a Monte Carlo experiment, reported in the following subsection}, this conjugate gradient method returns estimates that are consistently concentrated \mikko{near} the true \mikko{parameter values} (see also Figure \ref{fig:kernels_uncertainty}).

\begin{figure}[!t]
  \centering
  \includegraphics[width=0.78\linewidth]{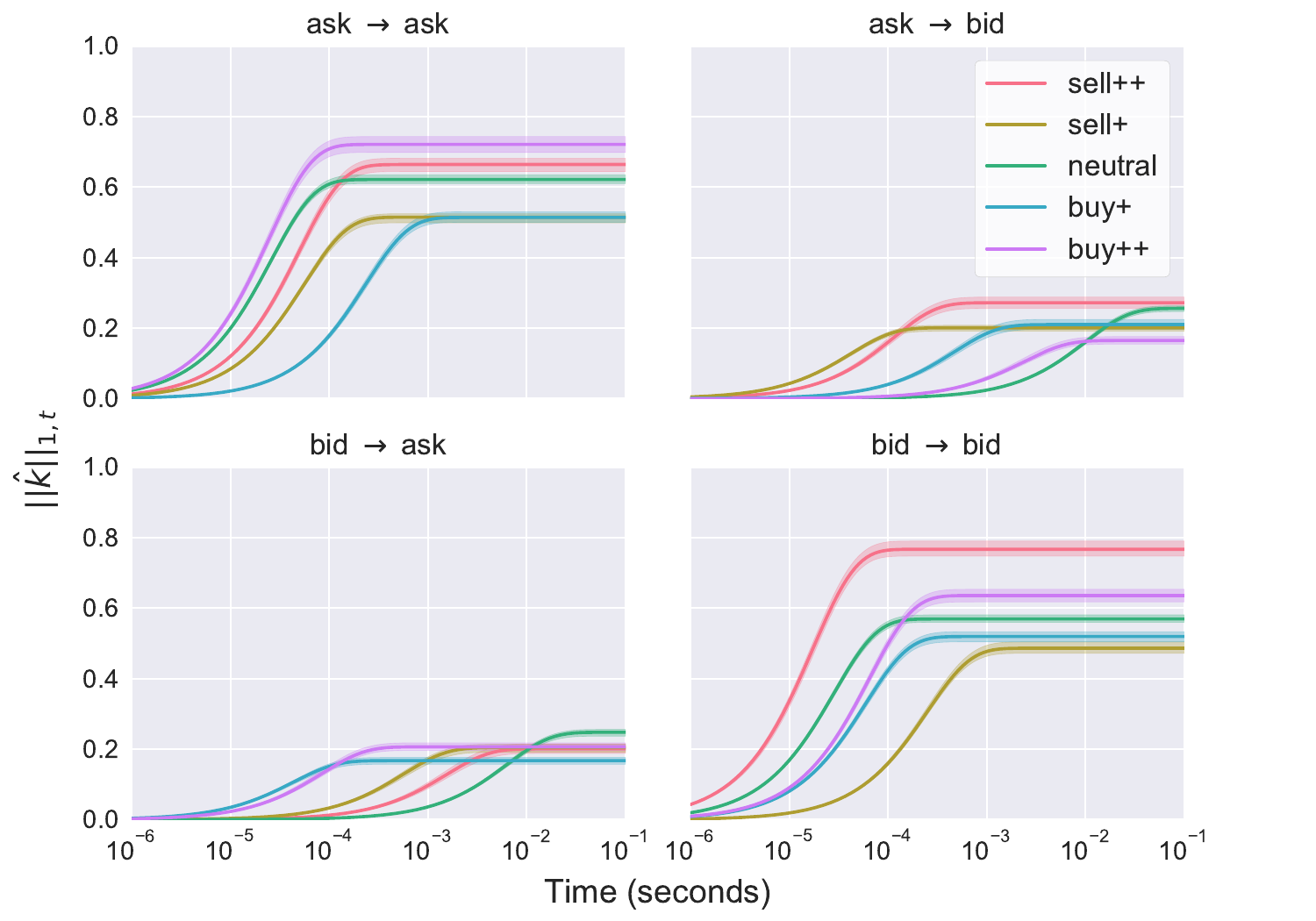}
  \caption{Uncertainty quantification of the \mikko{estimated excitation profiles for INTC on 13 February 2018 under $\mathrm{Model}_{\mathrm{QI}}$.} 
  \mikko{The parametric bootstrap procedure involves simulating} 100 paths \mikko{covering} 2.5 hours \mikko{of trading} using the \mikko{ML} estimate\mikko{s of the} parameters on the considered day and \mikko{applying ML} estimation \mikko{again to each of the simulated paths}. Three random sets of parameters are used as the initial guess in the optimisation procedure. We use the 100 estimates to compute a \maxime{99}\%-confidence interval for the \mikko{truncated kernel norm} (translucent area). \mikko{The solid line corresponds to the ML estimates using the original INTC data.}}
  \label{fig:kernels_uncertainty}
\end{figure}

\begin{table}[!p]
\maxime{
\tiny
\centering
\begin{subtable}[b]{0.99\textwidth}
\centering
\begin{tabular}{llllllllll}
\toprule
\multicolumn{2}{l}{$\impactCoeff{i}{1}{j}$} & \multicolumn{2}{l}{$\impactCoeff{i}{2}{j}$} & \multicolumn{2}{l}{$\impactCoeff{i}{3}{j}$} & \multicolumn{2}{l}{$\impactCoeff{i}{4}{j}$} & \multicolumn{2}{l}{$\impactCoeff{i}{5}{j}$} \\
\midrule
11,314 & 2,311           & 8,098 & 3,556           & 17,016  &  270          & 3,198  & 535         & 34,747  &  220         \\
139 & 33,460           & 724 & 4,497           & 178 &  15,497        & 3,778 &  7,943         & 2,353 &  12,452 \\
\bottomrule
\end{tabular}
\caption{Impact coefficients.}
\vspace*{2em}	
\end{subtable}
\begin{subtable}[b]{0.99\textwidth}
\centering
\begin{tabular}{llllllllll}
\toprule
\multicolumn{2}{l}{$\decayCoeff{i}{1}{j}$} & \multicolumn{2}{l}{$\decayCoeff{i}{2}{j}$} & \multicolumn{2}{l}{$\decayCoeff{i}{3}{j}$} & \multicolumn{2}{l}{$\decayCoeff{i}{4}{j}$} & \multicolumn{2}{l}{$\decayCoeff{i}{5}{j}$} \\
\midrule
17,480 &  9,563  & 15,187 & 19,637  & 29,890  & 1,456 &  6,493  & 3,071   & 45,546  & 1,097        \\
613 & 44,994  & 4,614  & 8,384  & 1,026  &  27,522  &  20,537 & 14,272 &  10,443 & 18,988 \\
\bottomrule   
\end{tabular}
\caption{Decay coefficients.}
\vspace*{2em}
\end{subtable}
\begin{subtable}[b]{0.49\textwidth}
\centering
\begin{tabular}{llllllllll}
\toprule
\multicolumn{5}{l}{$\transitionProbabilityIfEventFromAtoB{1}{i}{j}$}    & \multicolumn{5}{l}{$\transitionProbabilityIfEventFromAtoB{2}{i}{j}$}    \\
\midrule
.89 & .02 &  .05   & .03  & .01   &   .78   & .19 & .02 & <.01 & <.01 \\
.13 &  .83 &  .01 &  <.01 &  .02 & <.01 &  .75  & .2 &  .02 &  <.01 \\
.02 &  .17  & .79  & <.01 &  .02 &  .02 &  <.01 &  .79 &  .17 &  .02 \\
<.01 & .02 &  .21  & .75  & .01 & .02 &  <.01 &  <.01   & .82  & .14 \\
<.01 & <.01 &  .02 &  .19  & .78 & <.01 &  .03  & .06  & .02 &  .89 \\
\bottomrule
\end{tabular}
\caption{Transition distribution.}
\end{subtable}
\begin{subtable}[b]{0.49\textwidth}
\centering
\begin{tabular}{l}
\toprule
$\baseRate{i}$  \\
\midrule
2.2  \\
2.2 \\
\bottomrule
\end{tabular}
\caption{Base rate vector.}
\end{subtable}
\caption{Parameter values for Specification 1.}
\label{table:spec1}
}
\end{table}

\begin{table}[!p]
\maxime{
\tiny
\centering
\begin{subtable}[b]{0.49\textwidth}
\centering
\begin{tabular}{llllllllll}
\toprule
\multicolumn{2}{l}{$\impactCoeff{i}{1}{j}$} & \multicolumn{2}{l}{$\impactCoeff{i}{2}{j}$} & \multicolumn{2}{l}{$\impactCoeff{i}{3}{j}$} & \multicolumn{2}{l}{$\impactCoeff{i}{4}{j}$} & \multicolumn{2}{l}{$\impactCoeff{i}{5}{j}$} \\
\midrule
2           & 10           & 1            & 3           & 1,000         & 30          & 2,000          & 40         & 100         & 1,000         \\
10           & 2           & 3            & 1           & 20           & 3,000        & 2,000          & 50         & 60          & 2,000 \\
\bottomrule
\end{tabular}
\caption{Impact coefficients.}
\vspace*{2em}	
\end{subtable}
\hspace*{0.5em}
\begin{subtable}[b]{0.49\textwidth}
\centering
\begin{tabular}{llllllllll}
\toprule
\multicolumn{2}{l}{$\decayCoeff{i}{1}{j}$} & \multicolumn{2}{l}{$\decayCoeff{i}{2}{j}$} & \multicolumn{2}{l}{$\decayCoeff{i}{3}{j}$} & \multicolumn{2}{l}{$\decayCoeff{i}{4}{j}$} & \multicolumn{2}{l}{$\decayCoeff{i}{5}{j}$} \\
\midrule
10          & 15          & 8           & 4           & 3,000        & 500         & 6,000         & 160        & 500         & 8,000        \\
15          & 10          & 4           & 8           & 1000        & 5,000        & 10,000        & 300        & 120         & 5,000 \\
\bottomrule   
\end{tabular}
\caption{Decay coefficients.}
\vspace*{2em}
\end{subtable}
\begin{subtable}[b]{0.49\textwidth}
\centering
\begin{tabular}{llllllllll}
\toprule
\multicolumn{5}{l}{$\transitionProbabilityIfEventFromAtoB{1}{i}{j}$}    & \multicolumn{5}{l}{$\transitionProbabilityIfEventFromAtoB{2}{i}{j}$}    \\
\midrule
.7 & .3 & .0   & .0   & .0   & .0   & .1 & .2 & .3 & .4 \\
.1 & .8 & .1 & .0   & .0   & .2 & .1 & .4 & .2 & .1 \\
.0   & .1 & .6 & .2 & .1 & .1 & .3 & .1 & .3 & .2 \\
.2 & .2 & .3 & .1 & .2 & .0   & .0   & .1 & .8 & .1 \\
.1 & .3 & .3 & .1 & .2 & .1 & .0   & .1 & .1 & .7 \\
\bottomrule
\end{tabular}
\caption{Transition distribution.}
\end{subtable}
\begin{subtable}[b]{0.49\textwidth}
\centering
\begin{tabular}{l}
\toprule
$\baseRate{i}$  \\
\midrule
5  \\
1 \\
\bottomrule
\end{tabular}
\caption{Base rate vector.}
\end{subtable}
\caption{Parameter values for Specification 2.}
\label{table:spec2}
}
\end{table}

\subsubsection{Finite-sample performance} \label{subsec:test_estimation}

\mikko{We assess the finite-sample performance of the ML estimator in a small Monte Carlo experiment using a state-dependent Hawkes process with} exponential \mikko{kernel} of the form \eqref{eq:kernels_exp} when $\numberOfEvents=2$ and $\numberOfStates=5$. The parameters are naturally split into four groups:\ the transition probabilities \mikko{in} $\transitionProbabilities$, the base rates \mikko{in} $\baseRates$, the impact coefficients \mikko{in} $\impactCoeffs$ and the decay coefficients \mikko{in} $\decayCoeffs$. For \mikko{each} group of parameters $\collection{\theta_{i_j}}{j=1,\ldots,p_i}$ \mikko{and their estimator} $\collection{\hat\theta_{i_j}}{j=1,\ldots,p_i}$, \mikko{we} define the \emph{worst relative error} as
\begin{equation} \label{eq:worst_error}
	\varepsilon_{rel} := \frac{\hat\theta_{i_{j^\star}}-\theta_{i_{j^\star}}}{\theta_{i_{j^\star}}},\quad \mbox{where}\quad  j^\star = \argmax_{j} \frac{|\hat\theta_{i_{j}}-\theta_{i_{j}}|}{\theta_{i_{j}}}.
\end{equation}
For two different sets of parameter values and four different sample sizes, we estimate the distribution of $\varepsilon_{rel}$ for each of the four groups of parameters \mikko{using Monte Carlo via Algorithm \ref{alg:Ogata}}. \mikko{However, for the transition probabilities, we measure instead the worst \emph{absolute} error, defined by replacing the denominator in \eqref{eq:worst_error} by one.} The first set of parameter values (\mikko{S}pecification 1, \maxime{Table \ref{table:spec1}}) is constructed simply by averaging the daily estimates of $\mathrm{Model}_{\mathrm{QI}}$ using INTC data over the \maxime{19 trading days of February 2018}. The second set of parameter values (\mikko{S}pecification 2, \maxime{Table \ref{table:spec2}}) \mikko{is artificial}, \mikko{chosen to produce more drastic} changes in behaviour from one state to another. The results are displayed in Figure \ref{fig:convergence} and \mikko{indeed support the conjectured consistency of the ML estimators of the state-dependent Hawkes process.} \maxime{Note that the observed bimodality is due to the fact that the worst relative error inherently alternates between positive and negative values.}

\subsubsection{Uncertainty quantification via parametric bootstrap}

\mikko{The uncertainty of ML estimates of the parameters of a state-dependent Hawkes process can be quantified using a parametric bootstrap procedure. The procedure entails simulating realisations of the state-dependent Hawkes process under the estimated parameter values using Algorithm \ref{alg:Ogata} and then applying ML estimation again to each simulated realisation, producing a sample of estimates that approximates the distribution of the ML estimator. To exemplify the method, we apply it here to $\mathrm{Model}_{\mathrm{QI}}$ estimates for INTC on 13 February 2018. T}he results are presented in Figure \ref{fig:kernels_uncertainty}\mikko{, and we observe that the} uncertainty \mikko{of} the estimates is negligible compared to the \mikko{estimated excitation patterns}.

\section*{Acknowledgements}

Maxime Morariu-Patrichi gratefully acknowledges the Mini-DTC scholarship awarded by the Mathematics Department of Imperial College London. Mikko S. Pakkanen acknowledges partial support from CREATES (DNRF78), funded by the Danish National Research Foundation, \mikko{from EPSRC through the Platform Grant EP/I019111/1} and from the CFM--Imperial Institute of Quantitative Finance.

\bibliographystyle{apalike}
{\footnotesize
\bibliography{library}}

\end{document}